\documentclass[conference]{IEEEtran}

\usepackage{amsmath}
\usepackage{amsfonts}

\usepackage{amsthm}
\newtheorem{definition}{Definition}
\newtheorem{proposition}{Proposition}
\newtheorem{theorem}{Theorem}

\newtheorem{corollary}{Corollary}

\usepackage{tikz}
\usetikzlibrary{shapes.geometric}
\usetikzlibrary{decorations.markings}
\tikzset{-|->/.style={decoration={markings, mark=at position .5 with {\arrow{|}}, mark=at position 1 with {\arrow{>}}}, postaction={decorate}}}
\tikzset{->-/.style={decoration={markings, mark=at position .5 with {\arrow{>}}}, postaction={decorate}}}

\usepackage{xcolor}

\usepackage{hyperref}

\newcommand\define[1]{\textbf{#1}}

\newcommand\id{\operatorname{id}} 
\newcommand\op{\mathrm{op}} 
\newcommand\cokl{\operatorname{co-kl}} 
\newcommand\pto{\mathrel{\ooalign{\hfil$\mapstochar\mkern5mu$\hfil\cr$\longrightarrow$\cr}}} 

\newcommand\Set{\mathbf{Set}} 
\newcommand\Cat{\mathbf{Cat}} 
\newcommand\Lens{\mathbf{Lens}} 
\newcommand\Game{\mathbf{Game}} 

\newcommand\V{\mathbb V} 
\newcommand\K{\mathbb K} 
\newcommand\C{\mathbb C} 
\newcommand\St{\mathbb S} 
\newcommand\R{\mathbb R} 

\newcommand\G{\mathcal G} 
\renewcommand\H{\mathcal H} 
\newcommand\I{\mathcal I} 
\newcommand\D{\mathcal D} 

\newcommand\s{\mathfrak s} 
\renewcommand\t{\mathfrak t} 
\renewcommand\u{\mathfrak u} 
\renewcommand\a{\mathfrak a} 
\renewcommand\l{\mathfrak l} 
\renewcommand\r{\mathfrak r} 

\newcommand\B{\mathbf B} 

\title{Morphisms of open games}
\author{Jules Hedges}

\begin{document}

\maketitle

\begin{abstract}
	We define a notion of morphisms between open games, exploiting a surprising connection between lenses in computer science and compositional game theory.
	This extends the more intuitively obvious definition of globular morphisms as mappings between strategy profiles that preserve best responses, and hence in particular preserve Nash equilibria.
	We construct a symmetric monoidal double category in which the horizontal 1-cells are open games, vertical 1-morphisms are lenses, and 2-cells are morphisms of open games.
	States (morphisms out of the monoidal unit) in the vertical category give a flexible solution concept that includes both Nash and subgame perfect equilibria.
	Products in the vertical category give an external choice operator that is reminiscent of products in game semantics, and is useful in practical examples.
	We illustrate the above two features with a simple worked example from microeconomics, the market entry game.
\end{abstract}

\section{Introduction}

Open games provide a foundation to (economic) game theory that is strongly compositional.
In general open games are fragments of games that can be composed either sequentially or in parallel.
Key to this is the step of viewing open open games as the \emph{morphisms} of a symmetric monoidal category.
(This is an instance of the more general research programme of categorical open systems \cite{fong_algebra_open_interconnected_systems}.)

However, there are reasons that one might also wish to view open games as the objects of a category.
Most obviously, we would like to characterise certain open games using universal properties, in order to reason about them in a more abstract way.
In this paper we define a general notion of morphisms between open games, which we call \emph{contravariant lens morphisms}.
This heavily makes use of the factorisation of open games in terms of \emph{polymorphic lenses} introduced in \cite{hedges_coherence_lenses_open_games}.

We prove that open games and contravariant lens morphisms form a \emph{symmetric monoidal pseudo double category} \cite{shulman10}, the expected structure of 2-cells between arbitrary morphisms in a symmetric monoidal category.
Open games and globular morphisms (morphisms between open games of the same type) moreover form a symmetric monoidal bicategory, however we argue by example that the more general double-categorical structure is useful.

There are many possible inequivalent ways to define morphisms between open games, and we argue in favour of contravariant lens morphisms in particular by identifying two attractive properties.
Firstly \emph{states} of open games, that is morphisms out of the monoidal unit open game, provide a fully compositional solution concept that subsumes both Nash and subgame perfect equilibria in a flexible way.
Secondly, categorical products of open games correspond to an \emph{external choice} operator that is highly reminiscent of products in game semantics, and is useful in practical examples.
This paper culminates in a worked example that illustrates both of these aspects, in which we describe the market entry game, a simple but important example of game theory as applied to microeconomics.

An alternative definition of morphisms between open games is given in \cite{ghani_kupke_lambert_forsberg_compositional_treatment_iterated_open_games}, in order to characterise repeated games as final coalgebras of a functor prepending one additional stage.
Contravariant lens morphisms are inferior for this purpose, but on the other hand, the two attractive properties of contravariant lens morphisms identified are not shared by the morphisms of \cite{ghani_kupke_lambert_forsberg_compositional_treatment_iterated_open_games}.
This suggests that there is no single, canonical notion of morphisms between open games, but at least two with different useful properties.
It is expected, however, that all `reasonable' definitions of morphisms between open games will form a symmetric monoidal double category, and will agree on globular morphisms.


\section{Normal-form and extensive-form games}\label{sec:game-theory}

In this section we recall some basic definitions and results of game theory, which can be found in any standard textbook, such as \cite{fudenberg91}.

\begin{definition}
	An $n$-player \define{normal form game} consists of the following data:
	\begin{itemize}
		\item A sequence of sets $X_1, \ldots, X_n$ of \define{choices} for each player
		\item A \define{payoff function}
		\[ k : \prod_{i = 1}^n X_i \to \R^n \]
		giving a real-valued payoff for each player given choices by each player
	\end{itemize}
	A (pure) \define{strategy} for player $i$ in a normal form game is simply a choice $\sigma_i : X_i$, and a (pure) \define{strategy profile} $\sigma : \prod_{i = 1}^n X_i$ is a choice of strategy for each player.
\end{definition}

The adjective `pure' means `deterministic', in contrast with `mixed' or probabilistic strategies.
The equivalence of strategies and choices characterises a normal form game as both deterministic and simultaneous: in the probabilistic setting a strategy is a probability distribution over choices, and in a dynamic (sequential) game a strategy is a function from observations to choices.

Given a tuple $x : \prod_{i = 1}^n X_i$, we write $x_i : X_i$ for the $i$th projection,
\[ x_{-i} : \prod_{\substack{1 \leq j \leq n \\ j \neq i}} X_j \]
for the projection onto all but the $i$th component, and
\[ (x'_i, x_{-i}) : \prod_{j = 1}^n X_j \]
for the modification of $x$ with $i$th component $x'_i : X_i$.
(This notation is slightly imprecise, but is both useful and standard in game theory.)

\begin{definition}
	Given a normal form game $((X_i)_{i = 1}^n, q)$, let $\Sigma = \prod_{i = 1}^n X_i$ be its set of strategy profiles.
	We define the \define{best response relation} $\B \subseteq \Sigma \times \Sigma$ by $(\sigma, \sigma') \in \B$ iff for all players $1 \leq i \leq n$ and all $x_i : X_i$,
	\[ \left( k \left( \sigma'_i, \sigma_{-i} \right) \right)_i \geq \left( k \left( x_i, \sigma_{-i} \right) \right)_i \]
	
	A pure strategy profile is called a (pure) \define{Nash equilibrium} if it is a fixpoint of the best response relation, that is, if $(\sigma, \sigma) \in \B$.
\end{definition}

Equivalently, $\sigma$ is a Nash equilibrium iff for all players $1 \leq i \leq n$ and all \define{unilateral deviations} $x'_i : X_i$,
\[ (k (\sigma))_i \geq \left( k \left( x'_i, \sigma_{-i} \right) \right)_i \]
In words, a Nash equilibrium is a strategy profile in which no player can strictly increase their payoff by unilaterally deviating to another pure strategy.
(`Unilateral' means that the strategy profiles of all other players remain fixed.)


\begin{definition}
	An $n$-player game \define{extensive form game} is a tree where
	\begin{itemize}
		\item Each non-leaf node is labelled by a player, who makes the decision at that node
		\item Each leaf node is labelled by an $n$-tuple of real payoffs
		\item Each edge is labelled by an \define{action}, such that no two outgoing edges of a single node are labelled by the same action
		\item The nodes of the tree are partitioned into \define{information sets}
	\end{itemize}
	
	An information set is a nonempty subset of nodes such that
	\begin{itemize}
		\item Any pair of nodes in the same information set are labelled by the same player
		\item Any pair of nodes in the same information set have the same set of actions labelling their successors
	\end{itemize}
	
	A game of \define{perfect information} is an extensive form game whose information sets are all singletons.
\end{definition}

In section \ref{sec:backward-induction} we will define a simpler subclass of games of perfect information called \define{sequential games}.

\begin{definition}
	A \define{strategy} for a player $i$ in an extensive-form game is a map that takes each information set $x$ owned by player $i$, to a choice of action among successors of nodes in $x$.
	A \define{strategy profile} is a tuple of strategies for each player.
\end{definition}

A strategy profile induces a \define{play}, which is a path from the root node to a leaf node.
Plays are in bijection with leaf nodes, and so each strategy profile determines a payoff for each player.

\begin{definition}
	Each $n$-player extensive form game induces an $n$-player normal form game called its \define{normalisation} or \define{strategic form} as follows.
	The set of choices for player $i$ in the normalisation is defined to be the set of strategies for that player in the original game.
	The payoffs $k$ are determined by the play associated to the strategy profile.
	A \define{Nash equilibrium} of an extensive form game is a strategy profile that is a Nash equilibrium of its normalisation.
\end{definition}

\begin{definition}
	A \define{subgame} of an extensive form game $\G$ is a subtree $\H$ with the property that if $x$ is any node in $\H$, then any node in $\G$ in the same information set as $x$ is also in $\H$.
	The subtree $\H$ inherits the structure of an extensive form game from $\G$.
	
	A \define{subgame perfect equilibrium} of an extensive form game is a strategy profile that restricts to a Nash equilibrium on every subgame.
\end{definition}

(This definition implies that the root node of a subgame must be in a singleton information set.)

We will see an example illustrating these definitions in section \ref{sec:example}.

The remainder of this section considers a simplified special case of perfect information games, taken from \cite{escardo11,escardo12}.

\begin{definition}
	An $n$-player \define{sequential game} is an extensive-form game of perfect information in which at level $i$ of the tree all choices are made by player $i$, and in which any two nodes at the same level have the same set of actions available.
\end{definition}

The second condition implies that the tree of a sequential game is balanced.
An $n$-player sequential game is equivalently defined by sets $X_1, \ldots, X_n$ of actions, and a payoff function $k : \prod_{i = 1}^n X_i \to \R^n$.
Subgames are in bijection with \define{partial plays} $x_1, \ldots, x_{i - 1}$ for $1 \leq i \leq n$.
A strategy for player $i$ is a function
\[ \sigma_i : \prod_{j = 1}^i X_j \to X_i \]
and a strategy profile is a tuple
\[ \sigma : \prod_{i = 1}^n \left( \prod_{j = 1}^i X_j \to X_i \right) \]

\begin{definition}
	Let $(X_i)_{i = 1}^n$ be a sequence of sets for $n \geq 1$.
	Let $1 \leq p \leq q \leq n$.
	For a sequence
	\[ x : \prod_{j = 1}^{q - 1} X_i \]
	and a sequence of functions
	\[ \sigma : \prod_{i = p}^n \left( \prod_{j = 1}^{i - 1} X_j \to X_i \right) \]
	we define a sequence
	\[ v^\sigma_x : \prod_{i = 1}^n X_i \]
	extending $x$, called the \emph{strategic extension} of $x$ by $\sigma$, by the course-of-values recursion
	\[ (v^\sigma_x)_i = \begin{cases}
		x_i &\text{ if } i < q \\
		\sigma_i ((v^\sigma_x)_1, \ldots, (v^\sigma_x)_{i - 1}) &\text{ if } i \geq q
	\end{cases} \]
\end{definition}

With this notation, a strategy profile $\sigma$ of a sequential game is a Nash equilibrium iff
\[ \left( k \left( v^\sigma_{(v^\sigma)_1^{i - 1}, \sigma'_i ((v^\sigma)_1^{i - 1})} \right) \right)_i \geq \left( k \left( v^\sigma_{(v^\sigma)_1^{i - 1}, x_i} \right) \right)_i \]
for all players $1 \leq i \leq n$ and deviations $x_i : X_i$.
It is a subgame-perfect equilibrium iff for all players $1 \leq i \leq n$, all subgames $x_1, \ldots, x_{i - 1}$ and all deviations $x_i : X_i$,
\[ \left( k \left( v^\sigma_{x_1, \ldots, x_{i - 1}} \right) \right)_i \geq \left( k \left( v^\sigma_{x_1, \ldots, x_{i - 1}, x_i} \right) \right)_i \]

\section{The category of lenses}\label{sec:lenses}

In this section we recall and extend ideas from \cite{hedges_coherence_lenses_open_games} on lenses.

\begin{definition}
	Let $X, S, Y, R$ be sets. A \define{lens} $\lambda : (X, S) \to (Y, R)$ consists of a pair of functions $v_\lambda : X \to Y$, $u_\lambda : X \times R \to S$.
\end{definition}

We refer to such a pair $(X, S)$ as a \emph{diset}.
We write $\Phi, \Psi, \Theta$ to refer to disets.
This is not really a formal notion, but it carries a connotation that $X$ should be thought of `covariantly' and $S$ `contravariantly'.
(To say of a category $\mathcal C$ that its objects are disets is really to say that there is an identity-on-objects functor $\Set \times \Set^\op \to \mathcal C$.)

This definition of lens is called \emph{concrete lenses} in \cite{pickering_gibbons_wu_profunctor_optics}, and is essentially \emph{polymorphic lenses} in the absence of a polymorphic typesystem.
The relationship between this and the more familiar \emph{monomorphic lenses} is discussed in \cite{hedges_coherence_lenses_open_games}.

\begin{proposition}
	There is a category $\Lens$ whose objects are disets and whose morphisms are lenses.
	The identity lens $(X, S) \to (X, S)$ consists of the identity function $X \to X$ and the right projection $X \times S \to S$.
	The composition of $\lambda : (X, S) \to (Y, R)$ and $\mu : (Y, R) \to (Z, Q)$ is given by $v_{\mu \circ \lambda} = v_\mu \circ v_\lambda$ and $u_{\mu \circ \lambda} (x, q) = u_\lambda (x, u_\mu (v_\lambda (x), q))$.
\end{proposition}

\begin{proof}
	Routine.
\end{proof}

\begin{proposition}\label{prop:view-fibration}
	$v$ defines a fibration $\V : \Lens \to \Set$.
\end{proposition}

\begin{proof}
	\hyperref[prf:view-fibration]{See appendix.}
\end{proof}

For each set $X$ there is a left-multiplication comonad $(X \times) : \Set \to \Set$.
The co-kleisli category of this comonad has objects sets and morphisms
\[ \hom_{\cokl (X \times)} (R, S) = X \times R \to S \]

Every function $f : X \to Y$ induces an identity-on-objects functor $f^* : \cokl (Y \times) \to \cokl (X \times)$, where the maps
\[ f^* : \hom_{\cokl (Y \times)} (R, S) \to \hom_{\cokl (X \times)} (R, S) \]
are given by $f^* (u) (x, r) = u (f (x), r)$.
Thus we have a pseudofunctor $\cokl (- \times) : \Set^\op \to \Cat$.

\begin{proposition}\label{prop:fibres-of-view}
	$\V^{-1} (X) \cong \cokl (X \times)^\op$ for each set $X$, and for $f : X \to Y$ the reindexing functor $f^* : \V^{-1} (Y) \to \V^{-1} (X)$ is the opposite functor of that given above.
\end{proposition}

\begin{proof}
	\hyperref[prf:fibres-of-view]{See appendix.}
\end{proof}

The simple fibration $s (\Set) \to \Set$ is a fibration that plays a central role in the categorical semantics of simple type theory \cite[section 1.3]{jacobs-categorical-logic-type-theory}.
The category $s (\Set)$ has as objects pairs of sets, and as morphisms $(X, S) \to (Y, R)$ pairs of functions $X \to Y$ and $X \times S \to R$.

\begin{proposition}
	$\V$ is the fibrewise opposite of the simple fibration $s (\Set) \to \Set$.
\end{proposition}

\begin{proof}
	The fibre of $s (\Set)$ over $X$ is the kleisli category $\cokl (X \times)$ \cite[exercise 1.3.4]{jacobs-categorical-logic-type-theory}.
	Proving this is essentially the same as the above proof.
\end{proof}

\begin{proposition}
	There is an identity-on-objects functor $(-, -) : \Set \times \Set^\op \to \Lens$ defined by $v_{(f, g)} = f$ and $u_{(f, g)} = g \circ \pi_2$.
\end{proposition}

\begin{proof}
	Routine.
\end{proof}

Note that this functor is `almost' faithful, but there are counterexamples involving the empty set.
For example, take the unique function $f : 0 \to 1$ and the two functions $g_1 \neq g_2 : 1 \to 2 = \{ L, R \}$.
Then there is an equality of lenses $(f, g_1) = (f, g_2) : (0, 2) \to I$.
In \cite{hedges_coherence_lenses_open_games} the category of nonempty sets is identified with a subcategory of the category of lenses based on nonempty sets.

Next, we define a symmetric monoidal structure of $\Lens$.

\begin{definition}
	We define the monoidal product of disets to be $(X, S) \otimes (X', S') = (X \times X', S \times S')$, with monoidal unit $I = (1, 1)$.
	Given lenses $\lambda : \Phi \to \Psi$ and $\lambda' : \Phi' \to \Psi'$, we define a lens $\lambda \otimes \lambda' : \Phi \otimes \Phi' \to \Psi \otimes \Psi'$ by
	\begin{align*}
		v_{\lambda \otimes \lambda'} (x, x') &= (v_\lambda (x), v_{\lambda'} (x')) \\
		u_{\lambda \otimes \lambda'} ((x, x'), (r, r')) &= (u_\lambda (x, r), u_{\lambda'} (x', r'))
	\end{align*}
\end{definition}

We define the structure morphisms of $\Lens$ to be the image under $(-, -)$ of the corresponding structure morphisms of the cartesian monoidal category $\Set \times \Set^\op$.
That is to say,
\begin{align*}
	a_{(X, S), (X', S'), (X'', S'')} &= (a_{X, X', X''}, a_{S, S', S''}^{-1}) \\
	l_{(X, S)} &= (l_X, l_S^{-1}) \\
	r_{(X, S)} &= (r_X, r_S^{-1}) \\
	s_{(X, S), (X', S')} &= (s_{X, X'}, s_{S, S'}^{-1})
\end{align*}

\begin{proposition}
	$\Lens$ is a symmetric monoidal category.
\end{proposition}

\begin{proof}
	Instead of proving the Mac Lane axioms directly, we show that this monoidal structure results from a simpler symmetric monoidal structure on the fibres $\V^{-1} (X)$, by the Grothendieck construction for symmetric monoidal categories \cite[theorem 12.7]{shulman_framed_bicategories_monoidal_fibrations}.

	$\cokl (X)$ has finite products, given by cartesian products of sets.
	Consequently, cartesian products of sets give coproducts in $\cokl (X \times)^\op$.
	The reindexing functors $(f^*)^\op : \cokl (Y \times)^\op \to \cokl (X \times)^\op$ are identity-on-objects, and so preserve finite coproducts.
	This means in particular that the fibres $\V^{-1} (X)$ are symmetric monoidal and the reindexing functors are strong monoidal, and the base category $\Set$ is cartesian monoidal, and so the Grothendieck construction can be applied.
\end{proof}

\section{Open games and morphisms}\label{sec:morphisms}

\begin{proposition}
	There is a functor $\K : \Lens^\op \to \Set$, called the \define{continuation functor}, defined on disets by $\K (X, S) = X \to S$, and on lenses $\lambda : (X, S) \to (Y, R)$ by $\K (\lambda) (k) (x) = u_\lambda (x, k (v_\lambda (x)))$.
\end{proposition}

\begin{proof}
	Routine.
\end{proof}

\begin{definition}
	The \define{context functor} $\C : \Lens \times \Lens^\op \to \Set$ is defined by $\C (\Phi, \Psi) = \V (\Phi) \times \K (\Psi)$.
\end{definition}

Although we could see $\C$ as a profunctor $\Lens \to \Lens$, this will not be a helpful point of view.

\begin{proposition}
	$\hom_\Lens (I, (X, S)) \cong X$ defines a representation $\V \cong \hom_\Lens (I, -)$, and $\hom_\Lens ((X, S), I) \cong X \to S$ defines a representation $\K \cong \hom_\Lens (-, I)$.
\end{proposition}

\begin{proof}
	Routine.
\end{proof}

We summarise this by the slogan that in the category of lenses, \emph{states} (morphisms from $I$) are \emph{points}, and \emph{effects} (morphisms to $I$) are \emph{continuations}.

For convenience, we will immediately redefine $\V$ and $\K$ to be equal to these $\hom$-functors, so that we can write their action on lenses in terms of lens composition.

\begin{definition}
	Let $\Phi, \Psi$ be disets.
	An \define{open game} $\G : \Phi \pto \Psi$ consists of the following data:
	\begin{itemize}
		\item A set $\Sigma (\G)$ of \define{strategy profiles}
		\item For every $\sigma : \Sigma (\G)$, a lens $\G (\sigma) : \Phi \to \Psi$
		\item For every $c : \C (\Phi, \Psi)$, a \define{best response relation} $\B (c) \subseteq \Sigma (\G) \times \Sigma (\G)$
	\end{itemize}
\end{definition}

We call a pair $c = (h, k) : \C (\Phi, \Psi)$ a \define{context} for $\G$, where $h : \V (\Phi)$ is the \define{history} and $k : \K (\Psi)$ the \define{continuation}.
We also write $\s (\G) = \Phi$ and $\t (\G) = \Psi$ for the \define{source} and \define{target} of $\G$.

The equivalence between this definition and the more concrete definition in \cite[section 2.1.4]{hedges_towards_compositional_game_theory} (over the category of sets) is easy to see.
If $\Phi = (X, S)$ and $\Psi = (Y, R)$ then the family of lenses $\G (\sigma)$ is the same as the play and coplay functions $\Sigma (\G) \to (X \to Y)$, $\Sigma (\G) \to (X \times R \to S)$.
By the isomorphisms $\V (\Phi) \cong X$ and $\K (\Psi) \cong Y \to R$, the best response relation can equivalently be written as a function $X \times (Y \to R) \to (\Sigma (\G) \to \mathcal P (\Sigma (\G)))$, where $\mathcal P$ is powerset.

We note that given a pair of games $\G, \G' : \Phi \pto \Psi$ of the same type, there is an obvious way to define morphisms between them.

\begin{definition}
	Let $\G, \G' : \Phi \pto \Psi$ be open games.
	A \define{globular morphism} $\alpha : \G \to \G'$ is a function $\Sigma (\alpha) : \Sigma (\G) \to \Sigma (\G')$ such that
	\begin{itemize}
		\item $\G (\sigma) = \G (\Sigma (\alpha) (\sigma))$ for all $\sigma : \Sigma (\G)$
		\item If $(\sigma, \sigma') \in \B_\G (c)$ then $(\Sigma (\alpha) (\sigma), \Sigma (\alpha) (\sigma')) \in \B_{\G'} (c)$
	\end{itemize}
\end{definition}

It is expected that all reasonable definitions of general morphisms between open games will agree on the globular morphisms.
However it is demonstrated by example, both in this paper and in \cite{ghani_kupke_lambert_forsberg_compositional_treatment_iterated_open_games}, that the more general morphisms are necessary.
The morphisms we define in this paper are very different to those in \cite{ghani_kupke_lambert_forsberg_compositional_treatment_iterated_open_games}, but indeed agree on the globular morphisms.

The following definition is the key definition of this paper.

\begin{definition}
	Let $\G : \Phi \pto \Psi$ and $\G : \Phi' \pto \Psi'$ be open games.
	A \define{contravariant lens morphism} $\alpha : \G \to \G'$ consists of the following data:
	\begin{itemize}
		\item Lenses $\s (\alpha) : \Phi' \to \Phi$ and $\t (\alpha) : \Psi' \to \Psi$
		\item A function $\Sigma (\alpha) : \Sigma (\G) \to \Sigma (\G')$
	\end{itemize}
	satisfying the following two axioms:
	\begin{itemize}
		\item For all $\sigma : \Sigma (\G)$, the following diagram in $\Lens$ commutes:
		\begin{center} \begin{tikzpicture}[node distance=3cm, auto]
			\node (X) {$\Phi$}; \node (Y) [right of=X] {$\Psi$}; \node (X') [below of=X] {$\Phi'$}; \node (Y') [below of=Y] {$\Psi'$};
			\draw [->] (X) to node {$\G (\sigma)$} (Y) ; \draw [->] (X') to node {$\G' (\Sigma (\alpha) (\sigma))$} (Y');
			\draw [->] (X') to node [right] {$\s (\alpha)$} (X); \draw [->] (Y') to node [right] {$\t (\alpha)$} (Y);
		\end{tikzpicture} \end{center}
		\item For all $(h, k) : \C (\Phi', \Psi)$ and all $\sigma, \sigma' : \Sigma (\G)$, if
		\[ (\sigma, \sigma') \in \B_\G (\s (\alpha) \circ h, k) \]
		then
		\[ (\Sigma (\alpha) (\sigma), \Sigma (\alpha) (\sigma')) \in \B_{\G'} (h, k \circ \t (\alpha)) \]
	\end{itemize}
\end{definition}

We call a pair $(h, k) : \C (\Phi', \Psi)$ a \define{context} for $\alpha$.

We represent a contravariant lens morphism $\alpha : \G \to \G'$ as a square
\begin{center} \begin{tikzpicture}[auto]
	\node (X) at (0, 0) {$\Phi$}; \node (Y) at (3, 0) {$\Psi$}; \node (X') at (0, -3) {$\Phi'$}; \node (Y') at (3, -3) {$\Psi'$};
	\draw [-|->] (X) to node [above] {$\G$} node [below] {$\Sigma (\G)$} (Y);
	\draw [-|->] (X') to node [above] {$\G'$} node [below] {$\Sigma (\G')$} (Y');
	\draw [->] (X') to node [right] {$\s (\alpha)$} (X); \draw [->] (Y') to node [right] {$\t (\alpha)$} (Y);
	\draw [->] (1.5, -.5) to node {$\Sigma (\alpha)$} (1.5, -2.5);
\end{tikzpicture} \end{center}
anticipating the double category structure.

Having done the work of representing open games in terms of lenses, this definition is \emph{almost} automatic, but with one crucial twist: the lenses $\s (\alpha)$ and $\t (\alpha)$ go in the opposite direction to $\alpha$.
If they were covariant then the resulting definition would be similar to that of \cite{ghani_kupke_lambert_forsberg_compositional_treatment_iterated_open_games} and have many features in common, including the elegant representation of repeated games.
However, we will demonstrate in this paper that the alternative definition has several attractive features.
We give the definition the more specific name \emph{contravariant lens morphism} to distinguish it from alternatives, but since this definition is the subject of this paper, we will simply refer to it as a \emph{morphism} from now.

%

\begin{proposition}
	Open games and morphisms form a category $\Game_v$, with identities and composition lifted from $\Set$ and $\Lens$.
\end{proposition}

\begin{proof}
	Routine.
\end{proof}

The symbol $\Game_v$ is mnemonic for \emph{vertical}, hinting that this will be the vertical category of our double category.
(Note that we follow the orientation convention of \cite{shulman10}.)

\begin{proposition}
	Let $\G, \G' : \Phi \pto \Psi$ be open games.
	Then globular morphisms $\G \to \G'$ as previously defined are equivalent to morphisms $\alpha : \G \to \G'$ with $\s (\alpha) = \id_\Phi$ and $\t (\alpha) = \id_\Psi$.
\end{proposition}

\begin{proof}
	Trivial.
\end{proof}

\begin{proposition}\label{prop:multiple-fibration}
	~\newline \vspace{-.5cm} 
	\begin{itemize}
		\item $\s : \Game \to \Lens^\op$ is an opfibration whose opcartesian liftings are $\t$-vertical and $\Sigma$-vertical
		\item $\t : \Game \to \Lens^\op$ is a fibration whose cartesian liftings are $\s$-vertical and $\Sigma$-vertical
		\item $\Sigma : \Game \to \Set$ is a fibration whose cartesian liftings are $\s$-vertical and $\t$-vertical
	\end{itemize}
\end{proposition}

\begin{proof}
	\hyperref[prf:multiple-fibration]{See appendix.}
\end{proof}

\begin{proposition}
	There is a fibred forgetful functor from the fibration $\Sigma : \Game \to \Set$ to the family fibration $\mathbf{Fam} (\Lens) \to \Set$, that forgets best response.
\end{proposition}

\begin{proof}
	Routine.
\end{proof}

\section{The double category of open games}\label{sec:double-category}

A double category \cite{kelly_street_review_elements_2_categories} is defined as an internal category object in the category of large categories and functors.
(Compare that a 2-category is defined as a category \emph{enriched} over categories.)
Equivalently, a 2-category contains four sorts of things: objects, horizontal 1-cells, vertical 1-morphisms, and 2-cells.
Given a pair of horizontal 1-cells $F : X \pto Y$, $G : W \pto Z$, a 2-cell $\alpha : F \to G$ consists of a pair of vertical 1-morphisms $f : X \to W$, $g : Y \to X$ between the objects, and a square
\begin{center} \begin{tikzpicture}[node distance=3cm, auto]
	\node (X) {$X$}; \node (Y) [right of=X] {$Y$}; \node (W) [below of=X] {$W$}; \node (Z) [right of=W] {$Z$};
	\draw [-|->] (X) to node [above] {$G$} (Y); \draw [-|->] (W) to node [above] {$H$} (Z);
	\draw [->] (X) to node [right] {$f$} (W); \draw [->] (Y) to node [right] {$g$} (Z);
	\node at (1.5, -1.5) {$\Downarrow \alpha$};
\end{tikzpicture} \end{center}

In this paper we are concerned with pseudo double categories, in which the vertical morphisms form a category, but the horizontal morphisms form a category only up to invertible 2-cells.
Hence, there is a category of vertical morphisms and a bicategory of horizontal morphisms.

A standard example of a double category has as objects sets, horizontal 1-cells relations, vertical 1-cells functions and 2-cells inclusion.
Another has as objects categories, horizontal 1-cells profunctors, vertical 1-cells functors and 2-cells natural transformations.
We will show that there is a double category whose 1-cells are disets, horizontal 1-cells are open games, vertical 1-cells are reversed lenses, and 2-cells are contravariant lens morphisms.

In this paper we mostly follow the notation of \cite{shulman10}, which gives an explicit definition of symmetric monoidal pseudo double categories.

\begin{definition}
	Let $\Phi$ be a diset.
	We define an open game $\u (\Phi) : \Phi \pto \Phi$ by $\Sigma (\u (\Phi)) = 1$, $\u (\Phi) (*) = \id_\Phi$ and $(*, *) \in \B_{\u (\Phi)} (c)$ for all contexts $c : \C (\Phi, \Phi)$.
\end{definition}

\begin{proposition}
	For each lens $\lambda : \Phi \to \Psi$ there is a morphism of open games $\u (\lambda) : \u (\Psi) \to \u (\Phi)$ defined by $\s (\u (\lambda)) = \t (\u (\lambda)) = \lambda$ and $\Sigma (\u (\lambda)) = \id_1$.
	Then $\u$ defines a functor $\Lens^\op \to \Game_v$.
\end{proposition}

\begin{proof}
	Trivial.
\end{proof}

\begin{definition}
	Let $\G : \Phi \pto \Psi$ and $\H : \Psi \pto \Theta$ be open games.
	The open game $\H \odot \G : \Phi \pto \Theta$ is defined by
	\begin{itemize}
		\item $\Sigma (\H \odot \G) = \Sigma (\G) \times \Sigma (\H)$
		\item $(\H \odot \G) (\sigma, \tau) = \H (\tau) \circ \G (\sigma)$
		\item $((\sigma, \tau), (\sigma', \tau')) \in \B_{\H \odot \G} (h, k)$ iff
		\[ (\sigma, \sigma') \in \B_\G (h, k \circ \H (\tau)) \]
		and
		\[ (\tau, \tau') \in \B_\H (\G (\sigma) \circ h, k) \]
	\end{itemize}
\end{definition}

In the previous section we noted that the definition of open games in this paper is equivalent to that of \cite{hedges_towards_compositional_game_theory}.
Under this equivalence, the composition $\H \odot \G$ corresponds to the definition $\H \circ_N \G$ in \cite[sections 2.2.3 and 2.2.4]{hedges_towards_compositional_game_theory}, in which the continuation $k \circ \H (\tau)$ of $\G$ in the previous definition is referred to as $k_{\tau \circ}$.
This is a primitive form of sequential play, which will be illustrated in practice in the last section of this paper.

\begin{proposition}\label{prop:sequential-morphisms-exist}
	Let $\Phi \overset\G\pto \Psi \overset\H\pto \Theta$ and $\Phi' \overset{\G'}\pto \Psi' \overset{H'}\pto \Theta'$ be open games, and let $\alpha : \G \to \G'$ and $\beta : \H \to \H'$ be morphisms such that $\t (\alpha) = \s (\beta)$.
	Then there is a morphism $\beta \odot \alpha : \H \odot \G \to \H' \odot \G'$ defined by $\s (\beta \odot \alpha) = \s (\alpha)$, $\t (\beta \odot \alpha) = \t (\beta)$ and $\Sigma (\beta \odot \alpha) = \Sigma (\alpha) \times \Sigma (\beta)$.
\end{proposition}

\begin{proof}
	\hyperref[prf:sequential-morphisms-exist]{See appendix.}
\end{proof}

\begin{proposition}\label{prop:odot-bifunctor}
	$\odot$ defines a functor
	\[ \Game \times_{\Lens^\op} \Game \to \Game \]
	where the pullback is over
	\[ \Game \overset\s\longrightarrow \Lens^\op \overset\t\longleftarrow \Game \]
\end{proposition}

\begin{proof}
	\hyperref[prf:odot-bifunctor]{See appendix.}
\end{proof}

\begin{proposition}
	There are globular natural isomorphisms
	\begin{align*}
		\a_{\I, \H, \G} &: (\I \odot \H) \odot \G \overset\cong\longrightarrow \I \odot (\H \odot \G) \\
		\l_\G &: \u (\t (\G)) \odot \G \overset\cong\longrightarrow \G \\
		\r_\G &: \G \odot \u (\s (\G)) \overset\cong\longrightarrow \G
	\end{align*}
\end{proposition}

\begin{proof}
	Each of these morphisms is over the corresponding structure morphism of the cartesian monoidal category $\Set$.
\end{proof}

\begin{proposition}\label{thm:double-category}
	The above structures form a pseudo double category whose category of objects is $\Lens^\op$ and whose category of morphisms is $\Game_v$.
\end{proposition}

\begin{proof}
	\hyperref[prf:double-category]{See appendix.}
\end{proof}

It follows immediately that open games and globular morphisms form a bicategory.
The equivalence relation $\sim$ on games defined in \cite[section 2.2.2]{hedges_towards_compositional_game_theory} is precisely the relation of globular isomorphism, and the category $\Game (\Set)$ is precisely the category of horizontal morphisms modulo globular isomorphism.

\begin{definition}
	Let $\Phi_1, \Phi_2$ be disets.
	Since $\V (\Phi_1 \otimes \Phi_2) \cong \V (\Phi_1) \times \V (\Phi_2)$, we have projections
	\[ \V (\Phi_1) \xleftarrow{\pi_1} \V (\Phi_1 \otimes \Phi_2) \xrightarrow{\pi_2} \V (\Phi_2) \]
\end{definition}

\begin{definition}
	Let $\Phi, \Phi', \Psi, \Psi'$ be disets.
	We define functions
	\[ L : \hom_\Lens (\Phi', \Psi') \to (\C (\Phi \otimes \Phi', \Psi \otimes \Psi') \to \C (\Phi, \Psi)) \]
	\[ R : \hom_\Lens (\Phi, \Psi) \to (\C (\Phi \otimes \Psi', \Psi \otimes \Psi') \to \C (\Phi', \Psi')) \]
	by
	\[ L (\lambda') (h, k) = (\pi_1 (h), k \circ (\Psi \otimes (\lambda' \circ \pi_2 (h))) \circ r_\Psi^{-1}) \]
	\[ R (\lambda) (h, k) = (\pi_2 (h), k \circ ((\lambda \circ \pi_1 (h)) \otimes \Psi') \circ l_{\Psi'}^{-1}) \]
\end{definition}

These continuations are, more explicitly,
\[ \Psi \xrightarrow{r_\Psi^{-1}} \Psi \otimes I \xrightarrow{\Psi \otimes \pi_2 (h)} \Psi \otimes \Phi' \xrightarrow{\Psi \otimes \lambda'} \Psi \otimes \Psi' \xrightarrow{k} I \]
\[ \Psi' \xrightarrow{l_{\Psi'}^{-1}} I \otimes \Psi' \xrightarrow{\pi_1 (h) \otimes \Psi'} \Phi \otimes \Psi' \xrightarrow{\lambda \otimes \Psi'} \Psi \otimes \Psi' \xrightarrow{k} I \]

\begin{proposition}\label{prop:projection-lemma}
	Let $\Xi \xrightarrow\kappa \Phi \xrightarrow\lambda \Psi \xrightarrow\mu \Theta$ and $\Xi' \xrightarrow{\kappa'} \Phi' \xrightarrow{\lambda'} \Psi' \xrightarrow{\mu'} \Theta'$ be lenses.
	Then the following diagram commutes:
	\begin{center} \begin{tikzpicture}[node distance=3.5cm, auto]
		\node (A) {$\C (\Xi, \Theta)$}; \node (B) [right of=A] {$\C (\Xi \otimes \Xi', \Theta \otimes \Theta')$}; \node (C) [right of=B] {$\C (\Xi', \Theta')$};
		\node (D) [below of=A] {$\C (\Phi, \Psi)$}; \node (E) [below of=B] {$\C (\Phi \otimes \Phi', \Psi \otimes \Psi')$}; \node (F) [below of=C] {$\C (\Phi', \Psi')$};
		\draw [->] (B) to node [above=5pt] {$L (\mu' \circ \lambda' \circ \kappa')$} (A); \draw [->] (B) to node [above=5pt] {$R (\mu \circ \lambda \circ \kappa)$} (C);
		\draw [->] (E) to node [above] {$L (\lambda')$} (D); \draw [->] (E) to node {$R (\lambda )$} (F);
		\draw [->] (A) to node {$\C (\kappa, \mu)$} (D); \draw [->] (B) to node [left] {$\C (\kappa \otimes \kappa',$} node [right] {$\mu \otimes \mu')$} (E); \draw [->] (C) to node [left] {$\C (\kappa', \mu')$} (F);
	\end{tikzpicture} \end{center}
\end{proposition}

\begin{proof}
	\hyperref[prf:projection-lemma]{See appendix.}
\end{proof}

\begin{definition}
	Let $\G_1 : \Phi_1 \pto \Psi_1$ and $\G_2 : \Phi_2 \pto \Psi_2$ be open games.
	We define an open game $\G_1 \otimes \G_2 : \Phi_1 \otimes \Phi_2 \pto \Psi_1 \otimes \Psi_2$ by
	\begin{itemize}
		\item $\Sigma (\G_1 \otimes \G_2) = \Sigma (\G_1) \times \Sigma (\G_2)$
		\item $(\G_1 \otimes \G_2) (\sigma_1, \sigma_2) = \G_1 (\sigma_1) \otimes \G_2 (\sigma_2)$
		\item $((\sigma_1, \sigma_2), (\sigma_1', \sigma_2')) \in \B_{\G_1 \otimes \G_2} (c)$ iff
		\[ (\sigma_1, \sigma_1') \in \B_{\G_1} (L (\G_2 (\sigma_2)) (c)) \]
		and
		\[ (\sigma_2, \sigma_2') \in \B_{\G_2} (R (\G_1 (\sigma_1)) (c)) \]
	\end{itemize}
\end{definition}

Continuing the connection between definitions in this paper and those of \cite{hedges_towards_compositional_game_theory}, this definition corresponds to \cite[section 2.2.7]{hedges_towards_compositional_game_theory}, also denoted $\G_1 \otimes \G_2$.
It is a primitive form of simultaneous play.
The continuation parts of the contexts $L (\G_2 (\sigma_2)) ((h_1, h_2), k)$ and $R (\G_1 (\sigma_1)) ((h_1, h_2), k)$ are respectively referred to as $k_{\otimes \sigma_2 (h_2)}$ and $k_{\sigma_1 (h_1) \otimes}$.

\begin{proposition}\label{prop:tensor-morphisms}
	Let $\alpha_1 : \G_1 \to \G_1'$ and $\alpha_2 : \G_2 \to \G_2'$ be morphisms of open games.
	Then there is a morphism $\alpha_1 \otimes \alpha_2 : \G_1 \otimes \G_2 \to \G_1' \to \G_2'$ defined by $\s (\alpha_1 \otimes \alpha_2) = \s (\alpha_1) \otimes \s (\alpha_2)$, $\t (\alpha_1 \otimes \alpha_2) = \t (\alpha_1) \otimes \t (\alpha_2)$, and $\Sigma (\alpha_1 \otimes \alpha_2) = \Sigma (\alpha_1) \times \Sigma (\alpha_2)$.
\end{proposition}

\begin{proof}
	\hyperref[prf:tensor-morphisms]{See appendix.}
\end{proof}

\begin{proposition}
	$\otimes$ defines a functor $\Game_v \times \Game_v \to \Game_v$.
\end{proposition}

\begin{proof}
	Trivial.
\end{proof}

\begin{proposition}
	There are natural isomorphisms
	\begin{align*}
		\alpha_{\G, \H, \I} &: (\G \otimes \H) \otimes \I \overset\cong\longrightarrow \G \otimes (\H \otimes \I) \\
		\lambda_\G &: \u (I) \otimes \G \overset\cong\longrightarrow \G \\
		\rho_\G &: \G \otimes \u (I) \overset\cong\longrightarrow \G \\
		\sigma_{\G, \H} &: \G \otimes \H \overset\cong\longrightarrow \H \otimes \G
	\end{align*}
\end{proposition}

\begin{proof}
	These are the same as the structure morphisms in the symmetric monoidal category $\Lens^\op \times \Set \times \Lens^\op$, where $\Set$ is cartesian monoidal.
	It is routine to verify that these still natural isomorphisms in $\Game_v$.
\end{proof}

\begin{proposition}
	$\Game_v$ is a symmetric monoidal category.
\end{proposition}

\begin{proof}
	Routine.
\end{proof}

\begin{proposition}
	$\s, \t : \Game \to \Lens^\op$ are strict symmetric monoidal functors.
\end{proposition}

\begin{proof}
	Trivial.
\end{proof}

\begin{proposition}\label{prop:identitor-distributor-isomorphisms}
	There are globular isomorphisms
	\[ \mathfrak U_{X_1, X_2} : \u (X_1 \otimes X_2) \overset\cong\longrightarrow \u (X_1) \otimes \u (X_2) \]
	and
	\[ \mathfrak X_{\G_1, \G_2, \H_1, \H_2} : (\H_1 \otimes \H_2) \odot (\G_1 \otimes \G_2) \overset\cong\longrightarrow (\H_1 \odot \G_1) \otimes (\H_2 \odot \G_2) \]
\end{proposition}

\begin{proof}
	\hyperref[prf:identitor-distributor-isomorphisms]{See appendix.}
\end{proof}

\begin{theorem}\label{thm:symmetric-monoidal-double-category}
	Disets, open games and morphisms form a symmetric monoidal pseudo double category.
\end{theorem}

\begin{proof}
	\hyperref[prf:symmetric-monoidal-double-category]{See appendix.}
\end{proof}

Note that we cannot immediately apply \cite[theorem 5.1]{shulman10} to deduce that the horizontal bicategory is also symmetric monoidal, because the horizontal bicategory is not framed \cite{shulman_framed_bicategories_monoidal_fibrations}.
However, the hypotheses of this theorem appear to be stronger than necessary: all vertical isomorphisms, which includes the structure morphisms of the symmetric monoidal structure, have companions and conjoints.
This will not impact us in practice, because we will work directly with the double category rather than the bicategory structure.

A common (and reasonable) complaint has been that little game-theoretic information can be obtained from the string diagram alone.
For example, no information about the equilibria of a game can be understood from its string diagram.
Thus, despite having a graphical language, in order to answer any nontrivial question about an open game given its diagram it is still necessary to explicitly calculate the denotation, which is tedious and error-prone, and makes compositionality much less useful in practice.

It is expected that a monoidal double category supports a 3-dimensional graphical calculus of surface diagrams.
To completely formalise these claims appears to still require a large amount of foundational work, however.

Having obtained a higher-dimensional categorical language that is fully compositional and has the potential to be graphical, the next step is to formalise interesting game-theoretic concepts using this language.
This is the subject of the remainder of this paper.
The long-term aim towards which we are working is to be able to reason about game-theoretic concepts using nothing but surface diagrams.

\section{States of open games}\label{sec:states}

\begin{definition}
	Let $\G : \Phi \pto \Psi$ be an open game.
	A \define{state} of $\G$ is a pair $(\sigma, k)$ where $\sigma : \Sigma (\G)$ and $k : \K (\Psi)$ with the property that for all $h : \V (\Phi)$, $(\sigma, \sigma) \in \B_\G (h, k)$.
	We also say that $\sigma$ is a state of $\G$ \define{over} $k$.
	We write $\St (\G)$ for the set of states of $\G$.
\end{definition}

Informally, a state consists of a choice of continuation and a choice of strategy that is a Nash equilibrium for that continuation, for all histories.
As we will see, this combines aspects of Nash and subgame perfect equilibria in a flexible way.
The definition of states is a weakening of the definition of \emph{solution} considered in \cite[section 3.3]{hedges_towards_compositional_game_theory}, in which $\sigma$ may depend on $k$, but not on $h$.
The latter is more game-theoretically plausible (because many open games have interesting solutions in this sense), but is incompatible with $\otimes$ and hence not fully compositional.
The definition of state is intended as a compromise between game-theoretic usefulness and compositionality.

\begin{proposition}
	Let $\alpha : \G \to \G'$ be a morphism of open games, and let $\sigma$ be a state of $\G$ over $k$.
	Then $\Sigma (\alpha) (\sigma)$ is a state of $\G'$ over $\t (\alpha) \circ k$.
	Thus we have a function $\St (\alpha) : \St (\G) \to \St (\G')$.
\end{proposition}

\begin{proof}
	Let $h : \V (\s (\G'))$.
	Since $\sigma$ is a state of $\G$ over $k$, we have
	\[ (\sigma, \sigma) \in \B_\G (\s (\alpha) \circ h, k) \]
	Since $\alpha$ is a morphism, it follows that
	\[ (\Sigma (\alpha) (\sigma), \Sigma (\alpha) (\sigma)) \in \B_{\G'} (h, k \circ \t (\alpha)) \qedhere \]
\end{proof}

\begin{proposition}
	$\St$ defines a functor $\Game_v \to \Set$.
\end{proposition}

\begin{proof}
	Routine.
\end{proof}

\begin{theorem}
	$\St \cong \hom_{\Game_v} (\u (I), -)$.
\end{theorem}

\begin{proof}
	States $(\sigma, k)$ of $\G$ are in bijection with morphisms $\alpha : \u (I) \to \G$, as follows.
	There is a unique choice $\s (\alpha) = k \circ \G (\sigma)$ satisfying the first axiom.
	Since by definition $(*, *) \in \B_{\u (I)} (c)$ for all $c : \C (I, I)$, and all split contexts for $\alpha$ are of the form $(h, \id_I)$ for $h : \V (\Phi)$, the second axiom is equivalent to the condition that $(\sigma, \sigma) \in \B_\G (h, k)$ for all $h : \V (\Phi)$.
\end{proof}

This is summarised in the diagram
\begin{center} \begin{tikzpicture}[node distance=3cm, auto]
	\node (A) {$I$}; \node (B) [right of=A] {$I$}; \node (C) [below of=A] {$\Phi$}; \node (D) [below of=B] {$\Psi$};
	\draw [-|->] (A) to node [above] {$\u (I)$} node [below] {$1$} (B); \draw [-|->] (C) to node [above] {$\G$} node [below] {$\Sigma_\G$} (D);
	\draw [->] (C) to node {$k \circ \G (\sigma)$} (A); \draw [->] (D) to node [right] {$k$} (B);
	\draw [->] (1.5, -.5) to node {$\sigma$} (1.5, -2.5);
\end{tikzpicture} \end{center}

Noting that $\u (I)$ is the monoidal unit of $\Game_v$, states of open games are indeed states in the more general sense of monoidal categories.

\begin{corollary}[Backward induction]
	Let $\G$ and $\H$ be open games with $\t (\G) = \s (\H)$.
	Let $\tau$ be a state of $\H$ over $k$, and let $\sigma$ be a state of $\G$ over $k \circ \H (\tau)$.
	Then $(\sigma, \tau)$ is a state of $\H \odot \G$ over $k$.
\end{corollary}

As its name suggests, the previous result is similar to the \emph{backward induction} method of game theory, although in a far more general form.
The intuition is that we first fix the strategy profile $\tau$ of $\H$, and then players in $\G$ reason as though players in $\H$ play $\tau$.

\begin{corollary}
	Let $\G_1$ and $\G_2$ be open games.
	Let $\sigma_1$ be a state of $\G_1$ over $k_1$, and let $\sigma_2$ be a state of $\G_2$ over $k_2$.
	Then $(\sigma_1, \sigma_2)$ is a state of $\G_1 \otimes \G_2$ over $k_1 \otimes k_2$.
\end{corollary}

\begin{definition}
	An open game $\G : \Phi \pto \Psi$ is called \emph{strategically trivial} if it satisfies the following two properties:
	\begin{itemize}
		\item $\Sigma (\G) \cong 1$
		\item $(*, *) \in \B_\G (h, k)$ for all contexts $(h, k) : \C (\Phi, \Psi)$
	\end{itemize}
\end{definition}

Since a strategically trivial open game $\G : (X, S) \pto (Y, R)$ is defined up to globular isomorphism by its lens $\G (*) : (X, S) \to (Y, R)$, the (horizontal) subcategory of strategically trivial games is equivalent to $\Lens$.
Given a lens $\lambda : (X, S) \to (Y, R)$, we abuse notation and write $\lambda : (X, S) \pto (Y, R)$ for the corresponding strategically trivial open game.
In particular, for functions $f : X \to Y$ and $g : R \to S$ we have an open game $(f, g) : (X, S) \pto (Y, R)$, and for each set $X$ there is a strategically trivial open game $\varepsilon : (X, X) \pto I$ corresponding to the counit lens \cite{hedges_coherence_lenses_open_games} with update function $u_\varepsilon (x, *) = x$.

\begin{proposition}
	Let $\G : \Phi \pto \Psi$ be a strategically trivial game.
	Then for every continuation $k : \K (\Psi)$, $\G$ has exactly one state over $k$.
	In particular, a strategically trivial effect $\G : \Phi \pto I$ has a unique state.
\end{proposition}

\begin{proof}
	The unique morphism is
	\begin{center} \begin{tikzpicture}[node distance=3cm, auto]
		\node (A) {$I$}; \node (B) [right of=A] {$I$}; \node (C) [below of=A] {$\Phi$}; \node (D) [below of=B] {$\Psi$};
		\draw [-|->] (A) to node [above] {$\u (I)$} node [below] {$1$} (B); \draw [-|->] (C) to node [above] {$\G$} node [below] {$1$} (D);
		\draw [->] (C) to node {$k \circ \G (*)$} (A); \draw [->] (D) to node [right] {$k$} (B);
		\draw [double equal sign distance] (1.5, -.5) to (1.5, -2.5);
	\end{tikzpicture} \end{center}
	\vspace{-.75cm} \qedhere
\end{proof}

Note that the strategically trivial effects $(X, S) \pto I$ are `internal continuations', and are in bijection with functions $X \to S$.
These include the utility functions of game theory.

\begin{definition}
	For all sets $X$ and $Y$ we define an open game $\D_{X, Y} : (X, 1) \pto (Y, \R)$ (where $\R$ is the set of real numbers, interpreted as utility) called a \emph{decision}, as follows:
	\begin{itemize}
		\item $\Sigma (\D_{X, Y}) = X \to Y$
		\item $\D_{X, Y} (\sigma)$ is the unique lens $(X, 1) \to (Y, \R)$ with $\V (\D_{X, Y} (\sigma)) = \sigma$
		\item $(\sigma, \sigma') \in \D_{X, Y} (h, k)$ iff $k (\sigma' (h)) \geq k (y)$ for all $y : Y$
	\end{itemize}
\end{definition}

\begin{proposition}\label{prop:states-of-decisions}
	Let $X, Y$ be sets and $k : Y \to \R$.
	Then there is a bijective correspondence between
	\begin{itemize}
		\item States of $\D_{X, Y}$ over $k$
		\item Functions $X \to \arg\max (k)$, where $\arg\max (k) \subseteq Y$ is the set of maximising points of $k$
	\end{itemize}
	In particular, if $Y$ is finite then $\D_{X, Y}$ has at least one state over every $k$.
\end{proposition}

\begin{proof}
	The pair $(\sigma, k)$ is a state of $\D_{X, Y}$ iff $k (\sigma (x)) \geq k (y)$ for all $x : X$ and $y : Y$, which is equivalent to the range of $\sigma$ being $\arg\max (k)$.
\end{proof}

The following theorem and its proof are essentially restatements of \cite[sections 3.1.2 -- 3.1.4]{hedges_towards_compositional_game_theory}.
This is one of the key connections between open games and classical game theory, characterising tensor products of decisions as normal form games.

\begin{theorem}\label{thm:normal-form-theorem}
	Let $Y_1, \ldots, Y_n$ be an indexed family of sets for $n \geq 1$, and let $k : \prod_{i = 1}^n Y_i \to \R^n$ be a function.
	Then there is a bijective correspondence between
	\begin{itemize}
		\item States of $\bigotimes_{i = 1}^n \D_{1, Y_i}$ over $k$
		\item Pure Nash equilibria of the $n$-player normal form game with outcome function $k$
	\end{itemize}
\end{theorem}

\begin{proof}
	\hyperref[prf:normal-form-theorem]{See appendix.}
\end{proof}

Note that if $\alpha_1$ is a state of $\mathcal D_{I, Y_1}$ over $k_1$ and $\alpha_2$ is a state of $\D_{I, Y_2}$ over $k_2$ then $\alpha_1 \otimes \alpha_2$ is a state of $\mathcal D_{I, Y_1} \otimes \mathcal D_{I, Y_2}$ over $k_1 \otimes k_2$.
However, the continuation $k_1 \otimes k_2$ corresponds to a function $Y_1 \times Y_2 \to \R^2$ that is `separable', in the sense that $(k_1 \otimes k_2) (y_1, y_2) = (k_1 (y_1), k_2 (y_2))$, and hence the $\otimes$-separable states can only be equilibria of these separable games.
Such games are game-theoretically trivial in the sense that there is no strategic interaction between players, and so a game degenerates into a tuple of independent maximisation problems.
Only the $\otimes$-inseparable states of tensor products of decisions correspond to nontrivial Nash equilibria.

Despite the abstractions introduced in this paper the proof of the previous theorem is notably `manual', and so the reader might be left wondering what the purpose of the abstraction was.
The point of this theorem is that it allows us to talk about Nash equilibria purely internally in a categorical structure, which can be combined with the other abstractions we have introduced in order to later reason about Nash equilibria in purely abstract terms.

In the last section we will see both of these points illustrated: a Nash equilibrium of a bimatrix game must be introduced `atomically' since it cannot be composed from simpler components, but once it has been introduced, it can be further composed using purely categorical methods.

\section{Backward induction}\label{sec:backward-induction}

The previous theorem characterised Nash equilibria of normal form games in terms of states of open games.
In this section we do the same to Nash and subgame perfect equilibria of extensive form games of perfect information.

\begin{definition}
	Let $n \geq 1$ and let $X_1, \ldots, X_n$ be a sequence of sets.
	We define an open game
	\[ \D^\Delta_{X_1, \ldots, X_n} : \left( \prod_{i = 1}^{n - 1} X_i, \R^{n - 1} \right) \pto \left( \prod_{i = 1}^n X_i, \R^n \right) \]
	by the string diagram in figure \ref{fig:copy-decision}, where the black node is the (strategically trivial) copying open game
	\[ \left( \Delta_{\prod_{i = 1}^{n - 1} X_i}, 1 \right) : \left( \prod_{i = 1}^{n - 1} X_i, 1 \right) \pto \left( \prod_{i = 1}^{n - 1} X_i \times \prod_{i = 1}^{n - 1} X_i, 1 \right) \]
\end{definition}

\begin{figure}
	\begin{center} \begin{tikzpicture}
		\node [rectangle, minimum height=2cm, draw] (A) at (6, 0) {$\D_{\prod_{i = 1}^{n - 1} X_i, X_n}$};
		\node [circle, scale=.5, fill=black, draw] (m) at (3, .5) {};
		\node (X1) at (1.5, .5) {$\displaystyle \prod_{i = 1}^{n - 1} X_i$}; \node (R1) at (1.5, -.5) {$\R^{n - 1}$};
		\node (X2) at (9, 1.5) {$\displaystyle \prod_{i = 1}^{n - 1} X_i$}; \node (Xn) at (9, .5) {$X_n$}; \node (R2) at (9, -.5) {$\R^{n - 1}$}; \node (R) at (9, -1.5) {$\R$};
		\draw [->-] (X1) to (m); \draw [->-] (m) to [out=45, in=180] (X2); \draw [->-] (m) to [out=-45, in=180] (A); \draw [->-] (A.east |- Xn) to (Xn);
		\draw [->-] (R2) to [out=180, in=0] (6, -1.5) to [out=180, in=0] (R1); \draw [->-] (R) to [out=180, in=0] (A.east |- R2);
	\end{tikzpicture} \end{center}
	\caption{Definition of $\D^\Delta_{X_1, \ldots, X_n}$ in terms of $\D_{\prod_{i = 1}^{n - 1} X_i, X_n}$}
	\label{fig:copy-decision}
\end{figure}

It is routine to check that $\D^\Delta_{X_1, \ldots, X_n}$ is concretely given as follows, up to globular isomorphism.
The set of strategy profiles is
\[ \Sigma \left( \D^\Delta_{X_1, \ldots, X_n} \right) = \prod_{i = 1}^{n - 1} X_i \to X_n \]
Given $\sigma : \Sigma \left( \D^\Delta_{X_1, \ldots, X_n} \right)$, the lens $\D^\Delta_{X_1, \ldots, X_n} (\sigma)$ has
\[ v_{\D^\Delta_{X_1, \ldots, X_n} (\sigma)} : \prod_{i = 1}^{n - 1} X_i \to \prod_{i = 1}^n X_i, x \mapsto (x, \sigma (x)) \]
and
\[ u_{\D^\Delta_{X_1, \ldots, X_n} (\sigma)} : \prod_{i = 1}^{n - 1} X_i \times \R^n \to \R^{n - 1}, (x, r) \mapsto r_{-n} \]
For $h : \prod_{i = 1}^{n - 1} X_i$ and $k : \prod_{i = 1}^n X_i \to \R^n$, the best response relation
\[ (\sigma, \sigma') \in \B_{\D^\Delta_{X_1, \ldots, X_n}} (h, k) \]
holds iff
\[ k (h, \sigma' (h))_n \geq k (h, x_n)_n \]
for all $x_n : X_n$.

We now come to one of the key results of this paper: For an extensive form game, Nash equilibria correspond to states in general, and subgame perfect equilibria correspond to $\odot$-separable states.
Using this, we can use the category-theoretic notion of $\odot$-separability (which could moreover be visible in a surface diagram as a glued boundary) to \emph{define} the game-theoretic notion of subgame perfection for general open games.
We state and prove the theorem for the special case of finite sequential games (defined in section \ref{sec:game-theory}) for simplicity, but we will see an example in section \ref{sec:example} that is not in this special case.

\begin{theorem}\label{thm:extensive-form-theorem}
	Let $X_1, \ldots, X_n$ be an indexed family of sets for $n \geq 1$, and let $k : \prod_{i = 1}^n Y_i \to \R^n$.
	Then there is a bijective correspondence between:
	\begin{itemize}
		\item States of $\bigodot_{i = 1}^n \D^\Delta_{X_1, \ldots, X_i}$ over $k$
		\item Pure Nash equilibria of the $n$-player sequential game with outcome function $k$
	\end{itemize}
	Moreover there is a bijective correspondence between:
	\begin{itemize}
		\item States of $\bigodot_{i = 1}^n \D^\Delta_{X_1, \ldots, X_i}$ over $k$ that are of the form $\bigodot_{i = 1}^n \alpha_i$, where each $\alpha_i$ is a state of $\D^\Delta_{X_1, \ldots, X_i}$
		\item Subgame perfect equilibria of the $n$-player sequential game with outcome function $k$
	\end{itemize}
\end{theorem}

\begin{proof}
	\hyperref[prf:extensive-form-theorem]{See appendix.}
\end{proof}

The previous proof can be equivalently written as a \emph{proof by backward induction}, which is a proof by finite \emph{bar induction} on the tree of subgames.

In standard game theory, the \emph{representation} of a game is typically dissociated from its \emph{analysis}.
On any game of a given class (for example extensive-form games) one can define a variety of solution concepts (for example pure Nash equilibrium, mixed Nash equilibrium, subgame perfect equilibrium, correlated equilibrium, Bayesian Nash equilibrium).
A curious fact about open games is that the representation of a model as an open game automatically `builds in' a particular solution concept, namely pure Nash equilibrium.
This is because the $\odot$ and $\otimes$ operators make essential use of Nash-like reasoning, in which the players in each component act \emph{as though} they know which strategies will be played in the other component.

In particular, it has proven to be difficult to characterise subgame perfect equilibria, despite the fact that the \emph{representation} of open games appears to be dynamic (that is, has a temporal component; subgame perfect equilibrium is specifically intended for dynamic games).
A failed attempt was made in \cite{hedges_towards_compositional_game_theory} by modifying the operator $\odot$, there called $\circ_{SP}$, but it fails to form a symmetric monoidal category with $\otimes$, and hence does not support a string diagram language.
(See the conclusion section of \cite{hedges_towards_compositional_game_theory}.)

A direct solution is given in \cite{ghani_kupke_lambert_forsberg_compositional_treatment_iterated_open_games} by defining a `subgame perfection operator' that modifies an open game's solution concept from Nash to subgame perfect equilibrium.
From a practical perspective their solution is similar to the one in this paper, offering flexibility between Nash and subgame perfect equilibrium.
Ours has the specific advantage that it separates the representation of a game from its analysis, as in standard game theory.

%
%

\section{Limits and colimits of lenses}\label{sec:limits}

We will begin, for completeness of presentation, by considering limits of lenses.
However, we will not use them in this paper.

Let $\left< \V, \K \right> : \Lens \to \Set \times \Set^\op$ be the universal morphism
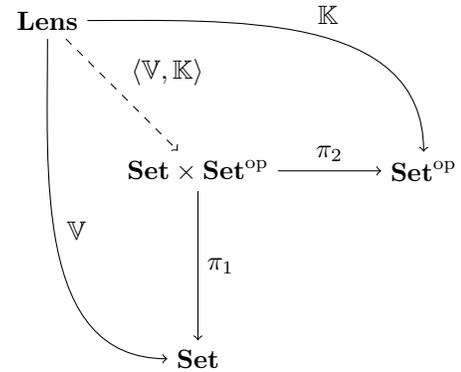
\begin{center} \begin{tikzpicture}[node distance=3cm, auto]
	\node (A) at (0, 2) {$\Lens$}; \node (B) at (2, 0) {$\Set \times \Set^\op$}; \node (C) at (2, -2.5) {$\Set$}; \node (D) at (5, 0) {$\Set^\op$};
	\draw [->] (B) to node {$\pi_1$} (C); \draw [->] (B) to node {$\pi_2$} (D);
	\draw [->] (A) to [out=-90, in=180] node {$\V$} (C); \draw [->] (A) to [out=0, in=90] node {$\K$} (D);
	\draw [->, dashed] (A) to node {$\left< \V, \K \right>$} (B);
\end{tikzpicture} \end{center}
Then $\left< \V, \K \right> \dashv (-, -)$, since the left adjoint acts on disets by
\[ \left< \V, \K \right> (X, S) = (X, X \to S) \]
and there are natural isomorphisms
\begin{align*}
	&\hom_{\Set \times \Set^\op} ((X, X \to S), (Y, R)) \\
	=\ &(X \to Y) \times (R \to (X \to S)) \\
	\cong\ &(X \to Y) \times (X \times R \to S) \\
	=\ &\hom_\Lens ((X, S), (Y, R)) \qedhere
\end{align*}

Since $(-, -)$ is a right adjoint, it preserves limits.
Since limits in $\Set \times \Set^\op$ are computed pointwise from limits and colimits in $\Set$, this allows us to compute various limits in $\Lens$.
For example, products in $\Lens$ are given by
\[ \prod_{i : I} (X_i, S_i) = \left( \prod_{i : I} X_i, \coprod_{i : I} S_i \right) \]

More generally, since $(X \times) : \Set \to \Set$ preserves colimits, $\cokl (X \times)$ has all limits, and so the fibres $\V^{-1} (X) \cong \cokl (X \times)^\op$ have all limits.
Moreover, we can show that the reindexing functors $f^* : \V^{-1} (Y) \to \V^{-1} (X)$ preserve limits.
We can then apply \cite[exercise 9.2.4]{jacobs-categorical-logic-type-theory} to deduce that $\Lens$ has all limits and $\V$ preserves them.

As an application of this, since $\Lens$ has pullbacks we can define a category of symmetric lenses as spans in $\Lens$.
This contrasts with monomorphic lenses, in which pullbacks do not exist in general unless the put-get law is imposed \cite{johnson_rosebrugh_spans_lenses}.
This can be used to give a definition of \emph{symmetric lawless lenses}.
It may be possible to apply this to develop a more symmetrical theory of open games, where the horizontal category $\Game$ is generalised from teleological \cite{hedges_coherence_lenses_open_games} to $\dag$-compact closed.

The simplest example of this is that any two objects over $0$ in $\Lens$ are uniquely isomorphic, and every object over $0$ is initial.

Next we come to colimits. 
For each set $S$, let $F_S : \Lens \to \Set$ be the functor given on objects by $F_S (Y, R) = Y \times (R \to S)$.
Then $(-, S) \dashv F_S$, since there are natural isomorphisms
\begin{align*}
	&\hom_\Lens ((X, S), (Y, R)) \\
	=\ &(X \to Y) \times (X \times R \to S) \\
	=\ &X \to Y \times (R \to S) \\
	=\ &\hom_\Set (X, F_S (Y, R)) \qedhere
\end{align*}

Since $(-, S)$ is a left adjoint, it preserves colimits.
In particular, $\Lens$ has coproducts of the form
\[ \coprod_{i : I} (X_i, S) = \left( \coprod_{i : I} X_i, S \right) \]
For example, $(1, 0)$ is terminal in $\Lens$.

\begin{proposition}\label{prop:product-games}
	The product in $\Game_v$ of a family of open games
	\[ \G_i : (X_i, S) \pto (Y_i, R) \]
	is
	\[ \prod_{i : I} \G_i : \left( \coprod_{i : I} X_i, S \right) \pto \left( \coprod_{i : I} Y_i, R \right) \]
	given as follows.
	The set of strategy profiles is
	\[ \Sigma \left( \prod_{i : I} \G_i \right) = \prod_{i : I} \Sigma (\G_i) \]
	The lenses are given by
	\[ \left( \prod_{i : I} \G_i \right) (\sigma) = \coprod_{i : I} \G_i (\sigma_i) \]
	Noting that
	\[ \V \left( \coprod_{i : I} X_i, S \right) \cong \coprod_{i : I} X_i \]
	the best response relation
	\[ (\sigma, \sigma') \in \B_{\prod_{i : I} \G_i} (\iota_j (h), k) \]
	holds iff
	\[ (\sigma_j, \sigma'_j) \in \B_{\G_j} (h, k \circ \iota_j) \]
	The projections $\pi_j : \prod_{i : I} \G_i \to \G_j$ are given by
	\begin{center} \begin{tikzpicture} [node distance=3cm, auto]
		\node (A) at (0, 0) {$\displaystyle \left( \coprod_{i : I} X_i, S \right)$}; \node (B) at (5, 0) {$\displaystyle \left( \coprod_{i : I} Y_i, R \right)$};
		\node (C) [below of=A] {$(X_j, S)$}; \node (D) [below of=B] {$(Y_j, R)$};
		\draw [-|->] (A) to node [above] {$\displaystyle \prod_{i : I} \G_i$} node [below] {$\displaystyle \prod_{i : I} \Sigma (\G_i)$} (B);
		\draw [-|->] (C) to node [above] {$\G_j$} node [below] {$\Sigma (\G_j)$} (D);
		\draw [->] (C) to node [right] {$(\iota_j, S)$} (A); \draw [->] (D) to node [right] {$(\iota_j, R)$} (B);
		\draw [->] (2.5, -1) to node {$\pi_j$} (2.5, -2.4);
	\end{tikzpicture} \end{center}
\end{proposition}

\begin{proof}
	\hyperref[prf:product-games]{See appendix.}
\end{proof}

The product $\G_1 \times \G_2$ is an \emph{external choice}, in which the decision of whether $\G_1$ or $\G_2$ is played is determined by the history.
This is strongly reminiscent of products in categories in game semantics \cite{abramsky_mccusker_game_semantics} and additive conjunction in linear logic.
In particular, it is typical of products in game semantics that \emph{strategies} compose by cartesian product, and \emph{plays} compose by disjoint union.
This analogy suggests the following correspondence between game semantics and compositional game theory: The \emph{player} of game semantics corresponds to the $n \geq 0$ noncooperative players in an open game, and the \emph{opponent} of game semantics corresponds to the context $(h, k)$.

\begin{proposition}\label{thm:product-states}
	Let $\G_i : (X_i, R) \pto (Y_i, S)$ be a family of open games.
	Then a state of $\prod_{i : I} \G_i$ consists, up to isomorphism, of a state for each $\G_i$.
\end{proposition}

\begin{proof}
	Since $\hom_{\Game_v} (\u (I), -)$ preserves limits,
	\[ \hom_{\Game_v} \left( \u (I), \prod_{i : I} \G_i \right) \cong \prod_{i : I} \hom_{\Game_v} (\u (I), \G_i) \]
\end{proof}

\section{Example: Market entry game}\label{sec:example}

We illustrate the previous three sections by considering the market entry game, a standard example of game theory in microeconomics \cite[example 9.B.3]{mas_collel_etal_microeconomic_theory}.
This game is played between two players: a prospective \emph{entrant} $E$ into a market, and an \emph{incumbent} $I$ in that market.
In the first stage, $E$ has the choice to either enter the market, or immediately quit.
If $E$ chooses to quit the game ends immediately.
If $E$ chooses to enter, in the second stage $E$ and $I$ simultaneously choose to either \emph{fight} or \emph{accommodate} the other (which could mean, for example, setting a low or high price).

The extensive form representation is depicted in figure \ref{fig:market-entry-game}.
The dotted line denotes that the two connected nodes share an information set, meaning that the firm $I$ does not know which of the two nodes they are in; in this way, the right hand subtree represents a simultaneous game.
(More precisely, the word `simultaneous' means the choices are made independently, but the physical timing might or might not be simultaneous.)

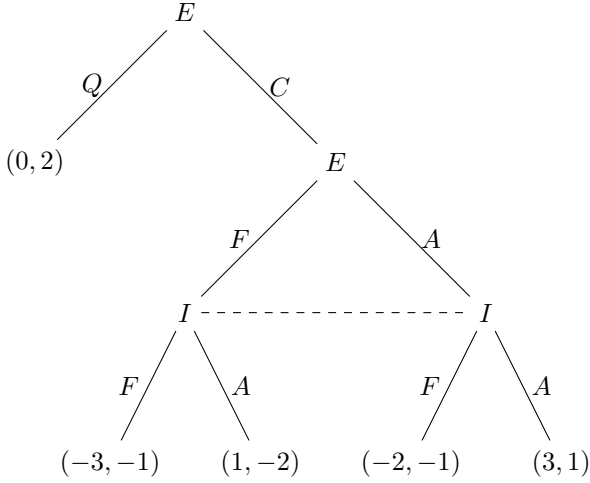
\begin{figure}
	\begin{center} \begin{tikzpicture}
		\node (A) at (2, 6) {$E$}; \node (B) at (0, 4) {$(0, 2)$}; \node (C) at (4, 4) {$E$};
		\node (D) at (2, 2) {$I$}; \node (E) at (6, 2) {$I$};
		\node (F) at (1, 0) {$(-3, -1)$}; \node (G) at (3, 0) {$(1, -2)$}; \node (H) at (5, 0) {$(-2, -1)$}; \node (I) at (7, 0) {$(3, 1)$};
		\draw [-] (A) to node [left] {$Q$} (B); \draw [-] (A) to node [right] {$C$} (C); \draw [-] (C) to node [left] {$F$} (D); \draw [-] (C) to node [right] {$A$} (E);
		\draw [-] (D) to node [left] {$F$} (F); \draw [-] (D) to node [right] {$A$} (G); \draw [-] (E) to node [left] {$F$} (H); \draw [-] (E) to node [right] {$A$} (I);
		\draw [-, dashed] (D) to (E);
	\end{tikzpicture} \end{center}
	\caption{Extensive form representation of market entry game}
	\label{fig:market-entry-game}
\end{figure}

The feature of this game that is specifically awkward for approaches based on theoretical computer science is that the extensive-form tree is unbalanced.
Other approaches use dependent types to allow types of choices to dependent on earlier values, but can still only allow more general `dependent subgames' such as in the market entry game using encoding tricks such as dummy moves, and using large negative utilities to rule out certain plays.
(Examples of game theory developed in a dependent type system include \cite{lescanne12,botta13}.)
An external choice operator solves the more general problem of dependent subgames in an elegant way.

We represent the market entry game as an open game as follows.
The decision of the first player (corresponding to the root node in figure \ref{fig:market-entry-game}) is modelled as a utility-maximising decision
\[ \D_{1, 1 + 1} : I \pto (1 + 1, \R) \]
where the choice $\iota_1 (*)$ represents not entering the market, and $\iota_2 (*)$ represents entering the market.

The left subgame is represented by
\[ (1, c_0) : (1, \R) \pto I \]
where $c_0 : 1 \to \R$ is the constant function with $c_0 (*) = 0$.
Notice that the utility of $2$ for player $I$ in the left subgame has a clear economic interpretation as the profit for the incumbent firm, but it plays no role game-theoretically, and it never appears in the open game representation.
(We could however add an additional utility-maximising player in the left subgame with only a single choice, which has no effect on the equilibrium analysis but is more faithful to the economic situation.)
As a strategically trivial effect, $(1, c_0)$ has a unique state $\alpha_L$, given by the unique lenses $\s (\alpha_L) : (1, \R) \to I$, $\t (\alpha_L) : I \to I$ and function $\Sigma (\alpha_L) : 1 \to 1$.

The right subgame is an open game
\[ \G : (1, \R) \pto I \]
which can be built compositionally from a pair of utility-maximising decisions $\D_{1, X} : I \pto (X, \R)$, where $X = \{ F, A \}$ is the set containing the choice to fight or accommodate, and the utility function $U : X \times X \to \R^2$.
This is represented by the string diagram in figure \ref{fig:market-entry-string}.

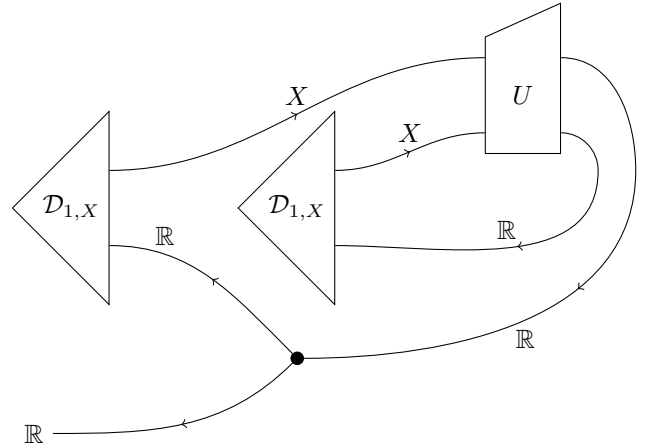
\begin{figure}
	\begin{center} \begin{tikzpicture}
		\node [isosceles triangle, isosceles triangle apex angle=90, shape border rotate=180, minimum width=2cm, draw] (D1) at (0, 0) {$\mathcal D_{1, X}$};
		\node [isosceles triangle, isosceles triangle apex angle=90, shape border rotate=180, minimum width=2cm, draw] (D2) at (-3, 0) {$\mathcal D_{1, X}$};
		\node [trapezium, trapezium left angle=0, trapezium right angle=75, shape border rotate=90, trapezium stretches=true, minimum height=1cm, minimum width=2cm, draw] (U) at (3, 1.5) {$U$};
		\node [circle, scale=0.5, fill=black, draw] (m) at (0, -2) {};
		\node (R) at (-3.5, -3) {$\R$};
		\node (d1) at (0, -.5) {}; \node (d2) at (0, .5) {}; \node (d3) at (0, 2.5) {}; \node (d4) at (0, 3.5) {}; \node (d5) at (0, 1) {}; \node (d6) at (0, 2) {};
		\draw [->-] (D1.east |- d2) to [out=0, in=180] node [above] {$X$} (U.west |- d5);
		\draw [->-] (D2.east |- d2) to [out=0, in=180] node [above] {$X$} (U.west |- d6);
		\draw [->-] (U.east |- d5) to [out=0, in=90] (4, .5) to [out=-90, in=0] node [above] {$\R$} (D1.east |- d1);
		\draw [->-] (U.east |- d6) to [out=0, in=90] (4.5, .5) to [out=-90, in=0] node [below] {$\R$} (m);
		\draw [->-] (m) to [out=135, in=0] node [above, near end] {$\R$} (D2.east |- d1); \draw [->-] (m) to [out=-135, in=0] (R);
	\end{tikzpicture} \end{center}
	\caption{String diagram for the right subgame $\G$}
	\label{fig:market-entry-string}
\end{figure}

By the coherence theorem for teleological categories \cite{hedges_coherence_lenses_open_games} this diagram defines $\G$ up to (in general non-globular) isomorphism.
Algebraically, in the terminology of this paper it is given by
\begin{align*}
	\varepsilon_\R &\odot (\u (\R, 1) \otimes \varepsilon_\R \otimes \u (1, \R)) \\
	&\odot (U, \id_{\R^2}) \\
	&\odot (\u (X, 1) \otimes (\id_1, \Delta_\R)) \\
	&\odot (\D_{1, X} \otimes \u (1, \R))
\end{align*}
Concretely, $\G$ is given as follows:
\begin{itemize}
	\item $\Sigma (\G) = X^2$
	\item $\G (\sigma) : (1, \R) \to I$ is the lens with $u_{\G (\sigma)} (*, *) = U (\sigma)_1$, i.e. the utility for player $E$ given strategy profile $\sigma$ for the subgame
	\item $(\sigma, \sigma') \in \B_\G (*, *)$ iff $U (\sigma'_1, \sigma_2)_1 \geq U (\overline{\sigma'_1}, \sigma_2)_1$ and $U (\sigma_1, \sigma'_2)_2 \geq U (\sigma_1, \overline{\sigma'_2})_2$, where $\overline{\ -\ } : X \to X$ gives the other choice, i.e. $\B_\G (*, *)$ is the best response function for the right simultaneous game
\end{itemize}


%

The right subgame has a unique Nash equilibrium $(A, A)$ in which both players accommodate.
Hence $\G$ has a unique state $\alpha_2$ with $\Sigma (\alpha) (*) = (A, A)$, in which the lens $\s (\alpha_2) : (1, \R) \to I$ has $u_{\s (\alpha_2)} (*, *) = 3$.
This is the payoff $U (A, A)_1$ for player $E$ in the Nash equilibrium.

The product $(1, c_0) \times \G : (1 + 1, \R) \pto (1 + 1, 1)$ is concretely given as follows:
\begin{itemize}
	\item $\Sigma ((1, c_0) \times \G) \cong X \times X$
	\item $((1, c_0) \times \G) (\sigma)$ is the lens $(1 + 1, \R) \to (1 + 1, 1)$ with view function
	\[ v_{((1, c_0) \times \G) (\sigma)} = \id_{1 + 1} \]
	and update function
	\[ u_{((1, c_0) \times \G) (\sigma)} (h, *) = \begin{cases}
		0 &\text{ if } h = \iota_1 (*) \\
		U (\sigma) &\text{ if } h = \iota_2 (*)
	\end{cases} \]
	\item The best response relation $(\sigma, \sigma') \in \B_{(1, c_0) \times \G} (h, k)$ iff either $h = \iota_1 (*)$, or $U (\sigma'_1, \sigma_2)_1 \geq U (\overline{\sigma'_1}, \sigma_2)_1$ and $U (\sigma_1, \sigma'_2)_2 \geq U (\sigma_1, \overline{\sigma'_2})_2$
\end{itemize}

Unlike the subgames $(1, c_0)$ and $\G$, the product $(1, c_0) \times \G$ cannot be seen as directly corresponding to a game in the classical sense.
By proposition \ref{thm:product-states}, $\left< \alpha_L, \alpha_2 \right>$ is the unique state of $(1, c_0) \times \G$, which is over the unique lens $\t (\left< \alpha_L, \alpha_R \right>) : (1 + 1, 1) \to I$.
It has $\Sigma (\left< \alpha_L, \alpha_R \right>) (*) = (A, A)$, and $\s (\left< \alpha_L, \alpha_R \right>) : (1 + 1, \R) \to I$ is the lens with
\[ u_{\s (\left< \alpha_L, \alpha_R \right>)} (h, *) = \begin{cases}
	0 &\text{ if } h = \iota_1 (*) \\
	3 &\text{ if } h = \iota_2 (*)
\end{cases} \]
Call this lens $\lambda$.

Since the previous function has a single maximising point, namely $\iota_2 (*)$, by proposition \ref{prop:states-of-decisions} the decision $\D_{1, 1 + 1}$ has a unique state $\alpha$ over $\lambda$.
This state has $\Sigma (\alpha) (*) : 1 \to 1 + 1$ given by $\Sigma (\alpha) (*) (*) = \iota_2 (*)$.

The open game
\[ \H = ((1, c_0) \times \G) \odot \D_{1, 1 + 1} : I \pto (1 + 1, 1) \]
will be our representation of the original market entry game.
It is concretely given, up to globular isomorphism, as follows.
The set of strategy profiles is
\[ \Sigma (\H) = (1 + 1) \times X^2 \]
consisting of a strategy for $E$ in the first round, and a strategy for both players in the subgame in which $E$ enters in the first round.
This is the same as the set of pure strategy profiles of the original extensive-form game.
(As in classical game theory, player $E$ is required to choose a contingent strategy for the second round, even if the strategy in the first round is to quit.)
$\H$ can be straightforwardly made into a scalar $I \pto I$ by postcomposing with the (unique) strategically trivial game $(1 + 1, 1) \pto I$.

The lens $\H (\sigma_1, \sigma_2, \sigma_3) : I \to (1 + 1, 1)$ has view function
\[ v_{\H (\sigma_1, \sigma_2, \sigma_3)} (*) = \sigma_1 \]

The best response relation $\B_\H (*, k)$ is the same as the best response relation for the market entry game given by classical game theory.
Concretely, the relation
\[ (\sigma, \sigma') \in \B_\H (*, k) \]
holds iff the following three conditions hold:
\begin{itemize}
	\item If $\sigma'_1 = \iota_1 (*)$ then $0 \geq U (\sigma_2, \sigma_3)_1$, and if $\sigma'_1 = \iota_2 (*)$ then $0 \leq U (\sigma_2, \sigma_3)_1$
	\item If $\sigma_1 = \iota_1 (*)$ then $U (\sigma'_2, \sigma_3)_1 \geq U (\overline{\sigma'_2}, \sigma_3)_1$
	\item If $\sigma_1 = \iota_1 (*)$ then $U (\sigma_2, \sigma'_3)_2 \geq U (\sigma_2, \overline{\sigma'_3})_2$
\end{itemize}

By the previous reasoning $\H$ has a unique $\odot$-separable state, namely $\alpha = \left< \alpha_L, \alpha_R \right> \odot \alpha_1$.
This corresponds to the unique subgame perfect equilibrium of the market entry game, namely that the entrant enters the market (choice $C$ or $\iota_2 (*)$), and then both players accommodate (choice $A$).
$\H$ has two additional states that are not $\odot$-separable, corresponding to the non-subgame-perfect Nash equilibria of the market entry game (or, equivalently, to the two additional fixpoints of $\B_\H (*, k)$).
One of these chooses the strategy profile $(\iota_1 (*), A, F)$, the other, $(\iota_1 (*), F, F)$.

\bibliographystyle{plainurl}
\bibliography{\string~/Dropbox/Work/refs}

\onecolumn
\appendix

\begin{proof}[Proof (proposition \ref{prop:view-fibration})]\label{prf:view-fibration}
	It is trivial that $\V$ is a functor.
	
	We define the following cleavage for $\V$.
	Given a function $f : X \to Y$ we define $f^* (Y, R) = (X, R)$.
	The cartesian lifting $\overline f (Y, R) : (X, R) \to (Y, R)$ is the lens defined by $v_{\overline f (Y, R)} = f$ and $u_{\overline f (Y, R)} (x, r) = r$.
	
	To show that $\overline f (Y, R)$ is indeed cartesian, let $\lambda : (W, T) \to (Y, R)$ be a lens such that $v_\lambda = f \circ g$ for some function $g : W \to X$.
	Define another lens $\mu : (W, T) \to (X, S)$ by $v_\mu = g$ and $u_\mu = u_\lambda$.
	Then $\overline f (Y, R) \circ \mu = \lambda$, because
	\[ v_{\overline f (Y, R) \circ \mu} = v_{\overline f (Y, R)} \circ v_\mu = f \circ g = v_\lambda \]
	and
	\[ u_{\overline f (Y, R) \circ \mu} (w, r) = u_\mu (w, u_{\overline f (Y, R)} (v_\mu (w), r)) = u_\lambda (w, r) \qedhere \]
\end{proof}

\begin{proof}[Proof (proposition \ref{prop:fibres-of-view})]\label{prf:fibres-of-view}
	The fibre $\V^{-1} (X)$ by definition has as objects bisets $(X, S)$ and as morphisms lenses $\lambda : (X, S) \to (X, R)$ with $v_\lambda = \id_X$.
	On the other hand objects of $\cokl (X \times)^\op$ are sets, and morphisms are
	\[ \hom_{\cokl (X \times)^\op} (S, R) = \hom_{\cokl (X \times)} (R, S) = X \times R \to S \]
	Therefore it can directly be seen that the functor $\V^{-1} (X) \to \cokl (X \times)^\op$ given by $(X, S) \mapsto S$ and $\lambda \mapsto u_\lambda$ is an isomorphism.
	
	By definition, the reindexing functor $f^* : \V^{-1} (Y) \to \V^{-1} (X)$ acts on objects by $f^* (Y, S) = (X, S)$, and takes a lens $\lambda : (Y, S) \to (Y, R)$ with $v_\lambda = \id_X$ to $f^* (\lambda) : (X, S) \to (X, R)$ with $v_{f^* (\lambda)} = \id_X$ and $u_{f^* (\lambda)} (x, r) = u_\lambda (f (x), r)$.
	This $f^* (\lambda)$ is the unique lens making the diagram
	\begin{center} \begin{tikzpicture}[node distance=3cm, auto]
		\node (A) {$(X, S)$}; \node (B) [right of=A] {$(Y, S)$}; \node (C) [below of=A] {$(X, R)$}; \node (D) [below of=B] {$(Y, R)$};
		\draw [->] (A) to node {$\overline f (Y, S)$} (B); \draw [->] (C) to node {$\overline f (Y, R)$} (D); \draw [->] (A) to node {$f^* (\lambda)$} (C); \draw [->] (B) to node {$\lambda$} (D);
	\end{tikzpicture} \end{center}
	commute.
	On the other hand, the functor $(f^*)^\op : \cokl (Y \times)^\op \to \cokl (X \times)^\op$ has maps
	\[ (f^*)^\op : \hom_{\cokl (Y \times)^\op} (S, R) \to \hom_{\cokl (X \times)^\op} (S, R) \]
	given by
	\[ (f^*)^\op (u) (x, r) = u (f (x), r) \]
	These can be directly seen to be equal under the isomorphism.
\end{proof}

\begin{proof}[Proof (proposition \ref{prop:multiple-fibration})]\label{prf:multiple-fibration}
	It is trivial that $\s$, $\t$ and $\Sigma$ are functors.
	We prove that $\s$ is an opfibration, with the others being similar.
	
	Let $\lambda : \Phi \to \Psi$ be a lens, so $\lambda : \hom_{\Lens^\op} (\Psi, \Phi)$, and let $\G$ be an open game with $\s (\G) = \Psi$.
	Let $\lambda_! (\G) : \Phi \pto \t (\G)$ be the open game defined by $\Sigma (\lambda_! (\G)) = \Sigma (\G)$, $\lambda_! (\G) (\sigma) = \G (\sigma) \circ \lambda$ and $\B_{\lambda_! (\G)} (c) = \B_\G (c)$.
	$\underline \lambda (\G) : \G \to \lambda_! (\G)$ is the morphism defined by $\s (\underline \lambda (\G)) = \lambda$, $\t (\underline \lambda (\G)) = \id_{\t (\G)}$ and $\Sigma (\underline \lambda (\G)) = \id_{\Sigma (\G)}$.
	The morphism axioms can be easily checked.
	$\underline \lambda (\G)$ is $\t$-vertical and $\Sigma$-vertical by construction.
	
	Let $\alpha : \G \to \H$ be a morphism such that $\s (\alpha) = \lambda \circ \mu$ for some lens $\mu : \s (\H) \to \Phi$ (hence $\s (\alpha) = \mu \circ \lambda$ for $\mu : \hom_{\Lens^\op} (\Phi, \s (\H))$).
	Then there is a morphism $\beta : \lambda_! (\G) \to \H$ defined by $\s (\beta) = \mu$, $\t (\beta) = \t (\alpha)$ and $\Sigma (\beta) = \Sigma (\alpha)$.
	The morphism axioms for $\beta$ follow from those for $\alpha$, together with $\lambda \circ \mu = \s (\alpha)$.
	Then
	\[ \s (\beta \circ \underline \lambda (\G)) = \s (\underline \lambda (\G)) \circ \s (\beta) = \lambda \circ \mu = \s (\alpha) \]
	and
	\[ \t (\beta \circ \underline \lambda (\G)) = \s (\underline \lambda (\G)) \circ \t (\beta) = \id_{\t (\G)} \circ \t (\alpha) = \t (\alpha) \]
	and
	\[ \Sigma (\beta \circ \underline \lambda (\G)) = \Sigma (\beta) \circ \Sigma (\underline \lambda (\G)) = \Sigma (\alpha) \circ \id_{\Sigma (\G)} = \Sigma (\alpha) \]
	Therefore $\beta \circ \underline \lambda (\G) = \alpha$.
	This situation is illustrated in figure \ref{fig:multiple-fibration}.
\end{proof}

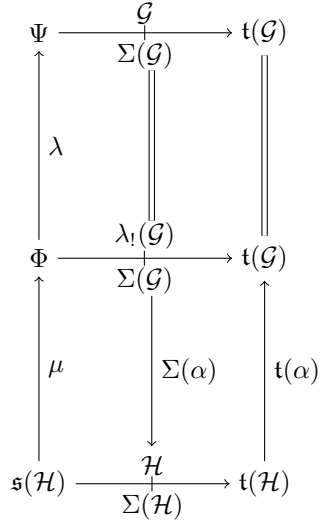
\begin{figure}
	\begin{center} \begin{tikzpicture}[node distance=3cm, auto]
		\node (Y) {$\Psi$}; \node (tG1) [right of=Y] {$\t (\G)$};
		\node (X) [below of=Y] {$\Phi$}; \node (tG2) [right of=X] {$\t (\G)$};
		\node (sH) [below of=X] {$\s (\H)$}; \node (tH) [right of=sH] {$\t (\H)$};
		\draw [-|->] (Y) to node [above] {$\G$} node [below] {$\Sigma (\G)$} (tG1);
		\draw [-|->] (X) to node [above] {$\lambda_! (\G)$} node [below] {$\Sigma (\G)$} (tG2);
		\draw [-|->] (sH) to node [above] {$\H$} node [below] {$\Sigma (\H)$} (tH);
		\draw [->] (X) to node [right] {$\lambda$} (Y); \draw [double equal sign distance] (tG2) to (tG1);
		\draw [->] (sH) to node [right] {$\mu$} (X); \draw [->] (tH) to node [right] {$\t (\alpha)$} (tG2);
		\draw [double equal sign distance] (1.5, -.5) to (1.5, -2.5); \draw [->] (1.5, -3.5) to node {$\Sigma (\alpha)$} (1.5, -5.5);
	\end{tikzpicture} \end{center}
	\caption{Illustration of proposition \ref{prop:multiple-fibration}}
	\label{fig:multiple-fibration}
\end{figure}

\begin{proof}[Proof (proposition \ref{prop:sequential-morphisms-exist})]\label{prf:sequential-morphisms-exist}
	The first axiom can be checked by diagram pasting:
	
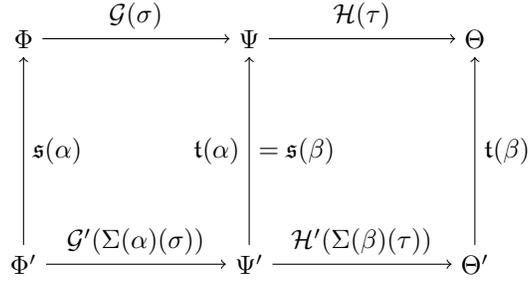
\begin{center} \begin{tikzpicture}[node distance=3cm, auto]
		\node (X) {$\Phi$}; \node (Y) [right of=X] {$\Psi$}; \node (Z) [right of=Y] {$\Theta$};
		\node (X') [below of=X] {$\Phi'$}; \node (Y') [below of=Y] {$\Psi'$}; \node (Z') [below of=Z] {$\Theta'$};
		\draw [->] (X) to node {$\G (\sigma)$} (Y); \draw [->] (Y) to node {$\H (\tau)$} (Z);
		\draw [->] (X') to node {$\G' (\Sigma (\alpha) (\sigma))$} (Y'); \draw [->] (Y') to node {$\H' (\Sigma (\beta) (\tau))$} (Z');
		\draw [->] (X') to node [right] {$\s (\alpha)$} (X); \draw [->] (Y') to node [left] {$\t (\alpha)$} node [right] {$= \s (\beta)$} (Y); \draw [->] (Z') to node [right] {$\t (\beta)$} (Z);
	\end{tikzpicture} \end{center}
	
	For the second axiom, let $(h, k) : \C (\Phi', \Theta)$ be a context for $\beta \odot \alpha$ and let $\sigma, \sigma' : \Sigma_\G$ and $\tau, \tau' : \Sigma_\H$ be strategy profiles.
	Suppose
	\[ ((\sigma, \tau), (\sigma', \tau')) \in \B_{\H \odot \G} (\s (\alpha) \circ h, k) \]
	so
	\[ (\sigma, \sigma') \in \B_\G (\s (\alpha) \circ h, k \circ \H (\tau)) \]
	and
	\begin{align*}
		(\tau, \tau') &\in \B_\H (\G (\sigma) \circ \s (\alpha) \circ h, k) \\
		&= \B_\H (\t (\alpha) \circ \G' (\Sigma (\alpha) (\sigma)) \circ h, k) \\
		&= \B_\H (\s (\beta) \circ \G' (\Sigma (\alpha) (\sigma)) \circ h, k)
	\end{align*}
	Then
	\begin{align*}
		(\Sigma (\alpha) (\sigma), \Sigma (\alpha) (\sigma')) &\in \B_{\G'} (h, k \circ \H (\tau) \circ \t (\alpha)) \\
		&= \B_{\G'} (h, k \circ \H (\tau) \circ \s (\beta)) \\
		&= \B_{\G'} (h, k \circ \t (\beta) \circ \H' (\Sigma (\beta) (\tau)))
	\end{align*}
	and
	\[ (\Sigma (\beta) (\tau), \Sigma (\beta) (\tau')) \in \B_{\H'} (\G' (\Sigma (\alpha) (\sigma)) \circ h, k \circ \t (\beta)) \]
	Therefore
	\[ ((\Sigma (\alpha) (\sigma), \Sigma (\beta) (\tau)), (\Sigma (\alpha) (\sigma'), \Sigma (\beta) (\tau'))) \in \B_{\H' \odot \G'} (h, k \circ \t (\beta)) \]
	and we are done.
\end{proof}

\begin{proof}[Proof (proposition \ref{prop:odot-bifunctor})]\label{prf:odot-bifunctor}
	This amounts to the following distributivity law: let $\Phi \overset\G\pto \Psi \overset\H\pto \Theta$, $\Phi' \overset{\G'}\pto \Psi' \overset{\H'}\pto \Theta'$ and $\Phi'' \overset{\G''}\pto \Psi'' \overset{\H''}\pto \Theta''$ be open games, and let $\G \xrightarrow\alpha G' \xrightarrow{\alpha'} \G''$ and $\H \xrightarrow\beta \H' \xrightarrow{\beta'} \H''$ be refinements such that $\t (\alpha) = \s (\beta)$ and $\ (\alpha') = \s (\beta')$.
	Then
	\[ (\beta' \odot \alpha') \circ (\beta \odot \alpha) = (\beta' \circ \beta) \odot (\alpha' \circ \alpha) \]
	
	For the source component,
	\begin{align*}
		\s ((\beta' \odot \alpha') \circ (\beta \odot \alpha)) &= \s (\beta' \odot \alpha') \circ \s (\beta \odot \alpha) \\
		&= \s (\alpha') \circ \s (\alpha) \\
		&= \s (\alpha' \circ \alpha) \\
		&= \s ((\beta' \circ \beta) \odot (\alpha' \circ \alpha))
	\end{align*}
	and similarly for the target component.
	For the best response component,
	\begin{align*}
		\Sigma ((\beta' \odot \alpha') \circ (\beta \odot \alpha)) (\sigma, \tau) =\ &\Sigma (\beta' \odot \alpha') (\Sigma (\beta \odot \alpha) (\sigma, \tau)) \\
		=\ &\Sigma (\beta' \odot \alpha') (\Sigma (\alpha) (\sigma), \Sigma (\beta) (\tau)) \\
		=\ &(\Sigma (\alpha') (\Sigma (\alpha) (\sigma)), \Sigma (\beta') (\Sigma (\beta) (\tau))) \\
		=\ &(\Sigma (\alpha' \circ \alpha) (\sigma), \Sigma (\beta' \circ \beta) (\tau)) \\
		=\ &(\Sigma ((\beta' \circ \beta) \odot (\alpha' \circ \alpha)) (\sigma, \tau)) \qedhere
	\end{align*}
\end{proof}

\begin{proof}[Proof (proposition \ref{thm:double-category})]\label{prf:double-category}
	It remains to prove that the following two diagrams (Mac Lane triangle and pentagon) commute:
	\begin{center} \begin{tikzpicture}[auto]
		\node (A) at (0, 0) {$(\H \odot \u (\t (\G))) \odot \G$}; \node (B) at (5, 0) {$\H \odot (\u (\t (\G)) \odot \G)$}; \node (C) at (2.5, -2) {$\H \odot \G$};
		\draw [->] (A) to node {$\a_{\G, \u_Y, \H}$} (B); \draw [->] (A) to node [left] {$\r_\H \odot \G$} (C); \draw [->] (B) to node {$\H \odot \l_\G$} (C);
	\end{tikzpicture} \begin{tikzpicture}[auto]
		\node (A) at (2.5, 0) {$((\mathcal J \odot \I) \odot H) \odot \G$};
		\node (B) at (0, -2.5) {$(\mathcal J \odot (\I \odot \H)) \odot \G$}; \node (C) at (0, -5) {$\mathcal J \odot ((\I \odot \H) \odot \G)$};
		\node (D) at (5, -2.5) {$(\mathcal J \odot \I) \odot (\H \odot \G)$}; \node (E) at (5, -5) {$\mathcal J \odot (\I \odot (\H \odot \G))$};
		\draw [->] (A) to node [left] {$\a_{\mathcal J, I, \H} \odot \G$} (B); \draw [->] (B) to node [left] {$\a_{\mathcal J, \I \odot \H, \G}$} (C);
		\draw [->] (A) to node {$\a_{\mathcal J \odot \I, \H, \G}$} (D); \draw [->] (D) to node {$\a_{\mathcal J, \I, \H \odot \G}$} (E);
		\draw [->] (C) to node {$\mathcal J \odot \a_{\I, \H, \G}$} (E);
	\end{tikzpicture} \end{center}
	Each composition path results in a globular morphism.
	Hence equality of the refinements follows from equality between the $\Sigma$ components, which are obtained from the corresponding Mac Lane axioms for the cartesian monoidal category $\Set$.
\end{proof}

\begin{proof}[Proof (proposition \ref{prop:projection-lemma})]\label{prf:projection-lemma}
	The first component
	\begin{center} \begin{tikzpicture}[node distance=3cm, auto]
		\node (A) {$\V (\Xi)$}; \node (B) [right of=A] {$\V (\Xi) \times \V (\Xi')$}; \node (C) [right of=B] {$\V (\Xi')$};
		\node (D) [below of=A] {$\V (\Phi)$}; \node (E) [below of=B] {$\V (\Phi) \times \V (\Phi')$}; \node (F) [below of=C] {$\V (\Phi')$};
		\draw [->] (B) to node [above=5pt] {$\pi_1$} (A); \draw [->] (B) to node [above=5pt] {$\pi_2$} (C);
		\draw [->] (E) to node [above] {$\pi_1$} (D); \draw [->] (E) to node {$\pi_2$} (F);
		\draw [->] (A) to node {$\V (\kappa)$} (D); \draw [->] (B) to node {$\V (\kappa) \times \V (\kappa')$} (E); \draw [->] (C) to node {$\V (\kappa')$} (F);
	\end{tikzpicture} \end{center}
	holds because $\pi_1, \pi_2$ are (isomorphic to) projections from a cartesian product of sets.
	
	For the second component, let $h : \V (\Xi \otimes \Xi')$.
	We reason that
	\[ (\Theta \otimes (\mu' \circ \lambda' \circ \kappa' \circ \pi_2 (h))) \circ r_\Psi^{-1} \circ \mu = (\Psi \otimes (\lambda' \circ \pi_2 ((\kappa \otimes \kappa') \circ h))) \circ r_\Psi^{-1} \circ (\mu \otimes \mu') \]
	from which commutativity of the left-hand square follows.
	(The proof for the right-hand square is symmetric.)
		
	See figure \ref{fig:projection-lemma}.
	Commutativity of the triangle labelled $(*)$ follows from chasing $h$ around the right-hand square of the first component.
	Commutativity of the top square is naturality of $r^{-1}$, and the other regions of the diagram commute by functorality of $\otimes$.
\end{proof}

\begin{figure*}
	\begin{center} \begin{tikzpicture}[node distance=4cm, auto]
		\node (A) {$\Psi$}; \node (B) [right of=A] {$\Theta$};
		\node (C) [below of=A] {$\Psi \otimes I$}; \node (D) [below of=B] {$\Theta \otimes I$};
		\node (E) at (-2, -8) {$\Psi \otimes \Phi'$}; \node (F) [right of=E] {$\Psi \otimes \Xi'$}; \node (G) [right of=F] {$\Theta \otimes \Xi'$};
		\node (H) at (0, -12) {$\Psi \otimes \Psi'$}; \node (I) [right of=H] {$\Theta \otimes \Theta'$};
		\draw [->] (A) to node {$\mu$} (B); \draw [->] (B) to node {$r_\Theta^{-1}$} (D); \draw [->] (A) to node {$r_\Psi^{-1}$} (C); \draw [->] (C) to node {$\mu \otimes I$} (D);
		\draw [->] (C) to node [left] {$\Psi \otimes \pi_2 ((\kappa \otimes \kappa') \circ h)$} (E); \draw [->] (C) to node {$\Psi \otimes \pi_2 (h)$} (F); \draw [->] (E) to node {$\Psi \otimes \kappa'$} (F);
		\draw [->] (D) to node {$\Theta \otimes \pi_2 (h)$} (G); \draw [->] (F) to node {$\mu \otimes \Xi'$} (G);
		\draw [->] (E) to node {$\Psi \otimes \lambda'$} (H); \draw [->] (F) to node {$\Psi \otimes (\lambda' \circ \kappa')$} (H);
		\draw [->] (G) to node {$\Theta \otimes (\mu' \circ \lambda' \circ \kappa')$} (I); \draw [->] (H) to node {$\mu \otimes \mu'$} (I);
		\node at (0, -6.5) {$(*)$};
	\end{tikzpicture} \end{center}
	\caption{Commuting diagram used to prove proposition \ref{prop:projection-lemma}}
	\label{fig:projection-lemma}
\end{figure*}

\begin{proof}[Proof (proposition \ref{prop:tensor-morphisms})]\label{prf:tensor-morphisms}
	For the first axiom, commutativity of the square
	\begin{center} \begin{tikzpicture}[node distance=3cm, auto]
		\node (X) at (0, 0) {$\s (\G_1) \otimes \s (\G_2)$}; \node (Y) at (6, 0) {$\t (\G_1) \otimes \t (\G_2)$};
		\node (X') [below of=X] {$\s (\G_1') \otimes \s (\G_2')$}; \node (Y') [below of=Y] {$\t (\G_1') \otimes \t (\G_2')$};
		\draw [->] (X) to node {$\G_1 (\sigma) \otimes \G_2 (\tau)$} (Y);
		\draw [->] (X') to node [above=5pt] {$\G_1' (\Sigma (\alpha_1) (\sigma)) \otimes \G_2' (\Sigma (\alpha_2) (\tau))$} (Y');
		\draw [->] (X') to node [right] {$\s (\alpha_1) \otimes \s (\alpha_2)$} (X); \draw [->] (Y') to node [left] {$\t (\alpha_1) \otimes \t (\alpha_2)$} (Y);
	\end{tikzpicture} \end{center}
	follows from the first axioms of $\alpha_1$ and $\alpha_2$.
	
	For the second axiom, let $c : \C (\s (\G_1' \otimes \G_2'), \t (\G_1 \otimes \G_2))$ be a context for $\alpha_1 \otimes \alpha_2$, and let $\sigma_1, \sigma_1' : \Sigma (\G_1)$ and $\sigma_2, \sigma_2' : \Sigma (\G_2)$ be strategy profiles such that
	\[ ((\sigma_1, \sigma_2), (\sigma_1', \sigma_2')) \in \B_{\G_1 \otimes \G_2} (\C (\s (\alpha_1 \otimes \alpha_2), \t (\G_1 \otimes \G_2)) (c)) \]
	Then
	\begin{align*}
		(\sigma_1, \sigma_1') \in\ &\B_{\G_1} (L (\G_2 (\sigma_2)) (\C (\s (\alpha_1 \otimes \alpha_2), \t (\G_1 \otimes \G_2)) (c))) \\
		=\ &\B_{\G_1} (\C (\s (\alpha_1), \t (\G_1)) (L (\G_2 (\sigma_2) \circ \s (\alpha_2)) (c))) \\
		=\ &\B_{\G_1} (\C (\s (\alpha_1), \t (\G_1)) (L (\t (\alpha_2) \circ \G_2' (\Sigma (\alpha_2) (\sigma_2))) (c)))
	\end{align*}
	and
	\begin{align*}
		(\sigma_2, \sigma_2') \in\ &\B_{\G_2} (R (\G_1 (\sigma_1)) (\C (\s (\alpha_1 \otimes \alpha_2), \t (\G_1 \otimes \G_2)) (c))) \\
		=\ &\B_{\G_2} (\C (\s (\alpha_2), \t (\G_2)) (R (\G_1 (\sigma_1) \circ \s (\alpha_1)) (c))) \\
		=\ &\B_{\G_2} (\C (\s (\alpha_2), \t (\G_2)) (R (\t (\alpha_1) \circ \G_1' (\Sigma (\alpha_1) (\sigma_1))) (c)))
	\end{align*}
	Therefore
	\begin{align*}
		(\Sigma (\alpha_1) (\sigma_1), \Sigma (\alpha_1) (\sigma_1')) \in\ &\B_{\G_1'} (\C (\s (\G_1'), \t (\alpha_1)) (L (\t (\alpha_2) \circ \G_2' (\Sigma (\alpha_2) (\sigma_2))) (c))) \\
		=\ &\B_{\G_1'} (L (\G_2' (\Sigma (\alpha_2) (\sigma_2))) (\C (\s (\G_1' \otimes  \G_2'), \t (\alpha_1 \otimes \alpha_2)) (c)))
	\end{align*}
	and
	\begin{align*}
		(\Sigma (\alpha_2) (\sigma_2), \Sigma (\alpha_2) (\sigma_2')) \in\ &\B_{\G_2'} (\C (\s (\G_2'), \t (\alpha_2)) (R (\t (\alpha_1) \circ \G_1' (\Sigma (\alpha_1) (\sigma_1))) (c))) \\
		=\ &\B_{\G_2'} (R (\G_1' (\Sigma (\alpha_1) (\sigma_1))) (\C (\s (\G_1' \otimes \G_2'), \t (\alpha_1 \otimes \alpha_2)) (c)))
	\end{align*}
	These combine to give
	\[ ((\Sigma (\alpha_1) (\sigma_1), \Sigma (\alpha_2) (\sigma_2)), (\Sigma (\alpha_1) (\sigma_1'), \Sigma (\alpha_2) (\sigma_2'))) \in \B_{\G_1' \otimes \G_2'} (\C (\s (\G_1' \otimes \G_2'), \t (\alpha_1 \otimes \alpha_2)) (c)) \]
	as required.
\end{proof}

\begin{proof}[Proof (proposition \ref{prop:identitor-distributor-isomorphisms})]\label{prf:identitor-distributor-isomorphisms}
	The first is trivial, with $\Sigma_{\mathfrak U_{X_1, X_2}} : 1 \to 1 \times 1, * \mapsto (*, *)$.

	For the second we take
	\begin{align*}
		\Sigma_\mathfrak X : (\Sigma (\G_1) \times \Sigma (\G_2)) \times (\Sigma (\H_1) \times \Sigma (\H_2)) &\to (\Sigma (\G_1) \times \Sigma (\H_1)) \times (\Sigma (\G_2) \times \Sigma (\H_2)), \\
		((\sigma_1, \sigma_2), (\tau_1, \tau_2)) &\mapsto ((\sigma_1, \tau_1), (\sigma_2, \tau_2))
	\end{align*}
	The first axiom follows from bifunctorality of $\otimes$ on $\Lens$:
	\[ (\H_1 (\tau_1) \otimes \H_2 (\tau_2)) \circ (\G_1 (\sigma_1) \otimes \G_2 (\sigma_2)) = (\H_1 (\tau_1) \circ \G_1 (\sigma_1)) \otimes (\H_2 (\tau_2) \circ \G_2 (\sigma_2)) \] 
	For the second axiom we calculate:
	\begin{align*}
		&(((\sigma_1, \sigma_2), (\tau_1, \tau_2)), ((\sigma_1', \sigma_2'), (\tau_1', \tau_2'))) \in \B_{(\H_1 \otimes \H_2) \odot (\G_1 \otimes \G_2)} (c) \\
		\iff\ &((\sigma_1, \sigma_2), (\sigma_1', \sigma_2')) \in \B_{\G_1 \otimes \G_2} (\mathbb C (\id_{X_1 \otimes X_2}, \H_1 (\tau_1) \otimes \H_2 (\tau_2)) (c)) \\
		\text{ and } &((\tau_1, \tau_2), (\tau_1', \tau_2')) \in \B_{\H_1 \otimes \H_2} (\mathbb C (\G_1 (\sigma_1) \otimes \G_2 (\sigma_2), \id_{Z_1 \otimes Z_2}) (c)) \\
		\iff\ &(\sigma_1, \sigma_1') \in \B_{\G_1} (L (\G_2 (\sigma_2)) (\mathbb C (\id_{X_1 \otimes X_2}, \H_1 (\tau_1) \otimes \H_2 (\tau_2)) (c))) \\
		\text{ and } &(\sigma_2, \sigma_2') \in \B_{\G_2} (R (\G_1 (\sigma_1)) (\mathbb C (\id_{X_1 \otimes X_2}, \H_1 (\tau_1) \otimes \H_2 (\tau_2)) (c))) \\
		\text{ and } &(\tau_1, \tau_1') \in \B_{\H_1} (L (\H_2 (\tau_2)) (\mathbb C (\G_1 (\sigma_1) \otimes \G_2 (\sigma_2), \id_{Z_1 \otimes Z_2}) (c))) \\
		\text{ and } &(\tau_2, \tau_2') \in \B_{\H_2} (R (\H_1 (\tau_1)) (\mathbb C (\G_1 (\sigma_1) \otimes \G_2 (\sigma_2), \id_{Z_1 \otimes Z_2}) (c))) \\
		\iff\ &(\sigma_1, \sigma_1') \in \B_{\G_1} (\mathbb C (\id_{X_1}, \H_1 (\tau_1)) (L (\H_2 (\tau_2) \circ \G_2 (\sigma_2)) (c))) \\
		\text{ and } &(\tau_1, \tau_1') \in \B_{\H_1} (\mathbb C (\G_1 (\sigma_1), \id_{Z_1}) (L (\H_2 (\tau_2) \circ \G_2 (\sigma_2)) (c))) \\
		\text{ and } &(\sigma_2, \sigma_2') \in \B_{\G_2} (\mathbb C (\id_{X_2}, \H_2 (\tau_2)) (R (\H_1 (\tau_1) \circ \G_1 (\sigma_1)) (c))) \\
		\text{ and } &(\tau_2, \tau_2') \in \B_{\H_2} (\mathbb C (\G_2 (\sigma_2), \id_{Z_2}) (R (\H_1 (\tau_1) \circ \G_1 (\sigma_1)) (c))) \\
		\iff\ &((\sigma_1, \tau_1), (\sigma_1', \tau_1')) \in \B_{\H_1 \odot \G_1} (L (\H_2 (\tau_2) \circ \G_2 (\sigma_2)) (c)) \\
		\text{ and } &((\sigma_2, \tau_2), (\sigma_2', \tau_2')) \in \B_{\H_2 \odot \G_2} (R (\H_1 (\tau_1) \circ \G_1 (\sigma_1)) (c)) \\
		\iff\ &(((\sigma_1, \tau_1), (\sigma_2, \tau_2)), ((\sigma_1', \tau_1'), (\sigma_2', \tau_2'))) \in \B_{(\H_1 \odot \G_1) \otimes (\H_2 \odot \G_2)} (c)
	\end{align*}
	The relevant equalities between contexts, which all follow from proposition \ref{prop:projection-lemma}, are summarised in figure \ref{fig:identitor-distributor-isomorphisms}.
\end{proof}

\begin{figure}
	\begin{center} \begin{tikzpicture} [node distance=6cm, auto]
		\node (CC) at (0, 0) {$\mathbb C (X_1 \otimes X_2, Z_1 \otimes Z_2)$};
		\node (CL) at (-6, 0) {$\mathbb C (X_1, Z_1)$}; \node (CR) at (6, 0) {$\mathbb C (X_2, Z_2)$};
		\node (TC) at (0, 3) {$\mathbb C (X_1 \otimes X_2, Y_1 \otimes Y_2)$}; \node (BC) at (0, -3) {$\mathbb C (Y_1 \otimes Y_2, Z_1 \otimes Z_2)$};
		\node (TL) [left of=TC] {$\mathbb C (X_1, Y_1)$}; \node (TR) [right of=TC] {$\mathbb C (X_2, Y_2)$};
		\node (BL) [left of=BC] {$\mathbb C (Y_1, Z_1)$}; \node (BR) [right of=BC] {$\mathbb C (Y_2, Z_2)$};
		\draw [->] (CC) to node {$L (\H_2 (\tau_2) \circ \G_2 (\sigma_2))$} (CL); \draw [->] (CC) to node {$R (\H_1 (\tau_1) \circ \G_1 (\sigma_1))$} (CR);
		\draw [->] (TC) to node {$L (\G_2 (\sigma_2))$} (TL); \draw [->] (TC) to node {$R (\G_1 (\sigma_1))$} (TR);
		\draw [->] (BC) to node {$L (\H_2 (\tau_2))$} (BL); \draw [->] (BC) to node {$R (\H_1 (\tau_1))$} (BR);
		\draw [->] (CC) to node [right] {$\mathbb C (\id_{X_1 \otimes X_2}, \H_1 (\tau_1) \otimes \H_2 (\tau_2))$} (TC); \draw [->] (CC) to node {$\mathbb C (\G_1 (\sigma_1) \otimes \G_2 (\sigma_2), \id_{Z_1 \otimes Z_2})$} (BC);
		\draw [->] (CL) to node [right] {$\mathbb C (\id_{X_1}, \H_1 (\tau_1))$} (TL); \draw [->] (CL) to node {$\mathbb C (\G_1 (\sigma_1), \id_{Z_1})$} (BL);
		\draw [->] (CR) to node [right] {$\mathbb C (\id_{X_2}, \H_2 (\tau_2))$} (TR); \draw [->] (CR) to node {$\mathbb C (\G_2 (\sigma_2), \id_{Z_2})$} (BR);
	\end{tikzpicture} \end{center}
	\caption{Instances of proposition \ref{prop:projection-lemma} used in proof of proposition \ref{prop:identitor-distributor-isomorphisms}}
	\label{fig:identitor-distributor-isomorphisms}
\end{figure}

\begin{proof}[Proof (theorem \ref{thm:symmetric-monoidal-double-category})]\label{prf:symmetric-monoidal-double-category}
	Following \cite{shulman10}, it remains to show that various diagrams commute.
	These are shown in figures \ref{fig:symmetric-monoidal-double-category-1}, \ref{fig:symmetric-monoidal-double-category-2}, \ref{fig:symmetric-monoidal-double-category-3} and \ref{fig:symmetric-monoidal-double-category-4}.
	These are all trivial to prove, by checking the $\s$, $\t$ and $\Sigma$ components separately.
\end{proof}

\begin{figure}
	\begin{center} \begin{tikzpicture}[node distance=3cm, auto]
		\node (A) at (0, 0) {$\u (\t (\G_1) \otimes \t (\G_2)) \odot (\G_1 \otimes \G_2)$}; \node (B) at (9, 0) {$(\u (\t (\G_1)) \otimes \u (\t (\G_2))) \odot (\G_1 \otimes \G_2)$};
		\node (C) [below of=A] {$\G_1 \otimes \G_2$}; \node (D) [below of=B] {$(\u (\t (\G_1)) \odot \G_1) \otimes (\u (\t (\G_2)) \odot \G_2)$};
		\draw [->] (A) to node {$\mathfrak U_{\t (\G_1), \t (\G_2)} \odot (\G_1 \otimes \G_2)$} (B); \draw [->] (A) to node {$\mathfrak l_{\G_1 \otimes \G_2}$} (C);
		\draw [->] (B) to node [left] {$\mathfrak X_{\G, \H, \u (\t (\G_1)), \u (\t (\G_2))}$} (D); \draw [->] (D) to node {$\mathfrak l_{\G_1} \otimes \mathfrak l_{\G_2}$} (C);
	\end{tikzpicture} \begin{tikzpicture}[node distance=3cm, auto]
		\node (A) at (0, 0) {$(\G_1 \otimes \G_2) \odot \u (\s (\G_1) \otimes \s (\G_2))$}; \node (B) at (9, 0) {$(\G_1 \otimes \G_2) \odot (\u (\s (\G_1)) \otimes \u (\s (\G_2)))$};
		\node (C) [below of=A] {$\G_1 \otimes \G_2$}; \node (D) [below of=B] {$(\G_1 \odot \u (\s (\G_1))) \otimes (\G_2 \odot \u (\s (\G_2)))$};
		\draw [->] (A) to node {$(\G_1 \otimes \G_2) \odot \mathfrak U_{\s (\G_1), \s (\G_2)}$} (B); \draw [->] (A) to node {$\mathfrak r_{\G_1 \otimes \G_2}$} (C);
		\draw [->] (B) to node [left] {$\mathfrak X_{\u (\s (\G_1)), \u (\s (\G_2)), \G_1, \G_2}$} (D); \draw [->] (D) to node {$\mathfrak r_{\G_1} \otimes \mathfrak r_{\G_2}$} (C);
	\end{tikzpicture} \begin{tikzpicture}[node distance=3cm, auto]
		\node (A) at (0, 0) {$((\I_1 \otimes \I_2) \odot (\H_1 \otimes \H_2)) \odot (\G_1 \otimes \G_2)$};
		\node (B) at (10, 0) {$((\I_1 \odot \H_1) \otimes (\I_2 \odot \H_2)) \odot (\G_1 \otimes \G_2)$};
		\node (C) [below of=A] {$(\I_1 \otimes \I_2) \odot ((\H_1 \otimes \H_2) \odot (\G_1 \otimes \G_2))$};
		\node (D) [below of=B] {$((\I_1 \odot \H_1) \odot \G_1) \otimes ((\I_2 \odot \H_2) \odot \G_2)$};
		\node (E) [below of=C] {$(\I_1 \otimes \I_2) \odot ((\H_1 \odot \G_1) \otimes (\H_2 \odot \G_1))$};
		\node (F) [below of=D] {$(\I_1 \odot (\H_1 \odot \G_1)) \otimes (\I_2 \odot (\H_2 \odot \G_2))$};
		\draw [->] (A) to node {$\mathfrak X_{\H_1, \H_2, \I_1, \I_2} \odot (\G_1 \otimes \G_2)$} (B); \draw [->] (E) to node {$\mathfrak X_{\H_1 \odot \G_1, \H_2 \odot \G_2, \I_1, \I_2}$} (F);
		\draw [->] (A) to node {$\mathfrak a_{\I_1 \otimes \I_2, \H_1 \otimes \H_2, \G_1 \otimes \G_2}$} (C); \draw [->] (C) to node {$(\I_1 \otimes \I_2) \odot \mathfrak X_{\G_1, \G_2, \H_1, \H_2}$} (E);
		\draw [->] (B) to node [left] {$\mathfrak X_{\G_1, \G_2, \I_1 \odot \H_1, \I_2 \odot \H_2}$} (D); \draw [->] (D) to node [left] {$\mathfrak a_{\I_1, \H_1, \G_1} \otimes \mathfrak a_{\I_2, \H_2, \G_2}$} (F);
	\end{tikzpicture} \end{center}
	\caption{Axioms for theorem \ref{thm:symmetric-monoidal-double-category}, part 1}
	\label{fig:symmetric-monoidal-double-category-1}
\end{figure}

\begin{figure}
	\begin{center} \begin{tikzpicture}[node distance=3cm, auto]
		\node (A) at (0, 0) {$\u ((X_1 \otimes X_2) \otimes X_3)$}; \node (B) at (7, 0) {$\u (X_1 \otimes (X_2 \otimes X_3))$};
		\node (C) [below of=A] {$\u (X_1 \otimes X_2) \otimes \mathfrak u (X_3)$}; \node (D) [below of=B] {$\u (X_1) \otimes \u (X_2 \otimes X_3)$};
		\node (E) [below of=C] {$(\u (X_1) \otimes \u (X_2)) \otimes \u (X_3)$}; \node (F) [below of=D] {$\u (X_1) \otimes (\u (X_2) \otimes \u (X_3))$};
		\draw [->] (A) to node {$\u (\alpha_{X_1, X_2, X_3})$} (B); \draw [->] (E) to node {$\alpha_{\u (X_1), \u (X_2), \u (X_3)}$} (F);
		\draw [->] (A) to node {$\mathfrak U_{X_1 \otimes X_2, X_3}$} (C); \draw [->] (C) to node {$\mathfrak U_{X_1, X_2} \otimes \u (X_3)$} (E);
		\draw [->] (B) to node {$\mathfrak U_{X_1, X_2 \otimes X_3}$} (D); \draw [->] (D) to node {$\u (X_1) \otimes \mathfrak U_{X_2, X_3}$} (F);
	\end{tikzpicture} \begin{tikzpicture}[node distance=3cm, auto]
		\node (A) at (0, 0) {$((\H_1 \otimes \H_2) \otimes \H_3) \odot ((\G_1 \otimes \G_2) \otimes \G_3)$};
		\node (B) at (10, 0) {$(\H_1 \otimes (\H_2 \otimes \H_3)) \odot (\G_1 \otimes (\G_2 \otimes \G_3))$};
		\node (C) [below of=A] {$((\H_1 \otimes \H_2) \odot (\G_1 \otimes \G_2)) \otimes (\H_3 \odot \G_3)$};
		\node (D) [below of=B] {$(\H_1 \odot \G_1) \otimes ((\H_2 \otimes \H_3) \odot (\G_2 \otimes \G_3))$};
		\node (E) [below of=C] {$((\H_1 \odot \G_1) \otimes (\H_2 \odot \G_2)) \otimes (\H_3 \odot \G_3)$};
		\node (F) [below of=D] {$(\H_1 \odot \G_1) \otimes ((\H_2 \odot \G_2) \otimes (\H_3 \odot \G_3))$};
		\draw [->] (A) to node {$\alpha_{\H_1, \H_2, \H_3} \odot \alpha_{\G_1, \G_2, \G_3}$} (B); \draw [->] (E) to node {$\alpha_{\H_1 \odot \G_1, \H_2 \odot \G_2, \H_3 \odot \G_3}$} (F);
		\draw [->] (A) to node {$\mathfrak X_{\G_1 \otimes \G_2, \G_3, \H_1 \otimes \H_2, \H_3}$} (C); \draw [->] (C) to node {$\mathfrak X_{\G_1, \G_2, \H_1, \H_2} \otimes (\H_3 \odot \G_3)$} (E);
		\draw [->] (B) to node [left] {$\mathfrak X_{\G_1, \G_2 \otimes \G_3, \H_1, \H_2 \otimes \H_3}$} (D); \draw [->] (D) to node [left] {$(\H_1 \odot \G_1) \otimes \mathfrak X_{\G_2, \G_3, \H_2, \H_3}$} (F);
	\end{tikzpicture} \end{center}
	\caption{Axioms for theorem \ref{thm:symmetric-monoidal-double-category}, part 2}
	\label{fig:symmetric-monoidal-double-category-2}
\end{figure}

\begin{figure}
		\begin{center} \begin{tikzpicture}[node distance=3cm, auto]
		\node (A) {$\u (I \otimes X)$}; \node (B) [right of=A] {$\u (I) \otimes \u (X)$}; \node (C) [below of=B] {$\u (X)$};
		\draw [->] (A) to node {$\mathfrak X_{I, X}$} (B); \draw [->] (B) to node {$\lambda_{\u (X)}$} (C); \draw [->] (A) to node [left] {$\u (\lambda_X)$} (C);
	\end{tikzpicture} \begin{tikzpicture}[node distance=3cm, auto]
		\node (A) at (0, 0) {$(\u (I) \otimes \H) \odot (\u (I) \otimes \G)$}; \node (B) at (6, 0) {$(\u (I) \odot \u (I)) \otimes (\H \odot \G)$};
		\node (C) [below of=A] {$\H \odot \G$}; \node (D) [below of=B] {$\u (I) \otimes (\H \odot \G)$};
		\draw [->] (A) to node {$\mathfrak X_{\u (I), \G, \u (I), \H}$} (B); \draw [->] (A) to node {$\lambda_\H \odot \lambda _\G$} (C);
		\draw [->] (B) to node {$\mathfrak l_{\u (I)} \otimes (\H \odot \G)$} (D); \draw [->] (D) to node {$\lambda_{\H \odot \G}$} (C);
	\end{tikzpicture} \begin{tikzpicture}[node distance=3cm, auto]
		\node (A) {$\u (X \otimes I)$}; \node (B) [right of=A] {$\u (X) \otimes \u (I)$}; \node (C) [below of=B] {$\u (X)$};
		\draw [->] (A) to node {$\mathfrak X_{X, I}$} (B); \draw [->] (B) to node {$\rho_{\u (X)}$} (C); \draw [->] (A) to node [left] {$\u (\rho_X)$} (C);
	\end{tikzpicture} \begin{tikzpicture}[node distance=3cm, auto]
		\node (A) at (0, 0) {$(\H \otimes \u (I)) \odot (\G \otimes \u (I))$}; \node (B) at (6, 0) {$(\H \odot \G) \otimes (\u (I) \odot \u (I))$};
		\node (C) [below of=A] {$\H \odot \G$}; \node (D) [below of=B] {$(\H \odot \G) \otimes \u (I)$};
		\draw [->] (A) to node {$\mathfrak X_{\G, \u (I), \H, \u (I)}$} (B); \draw [->] (A) to node {$\rho_\H \odot \rho_\G$} (C);
		\draw [->] (B) to node {$(\H \odot \G) \otimes \mathfrak r_{\u (I)}$} (D); \draw [->] (D) to node {$\rho_{\H \odot \G}$} (C);
	\end{tikzpicture} \end{center}
	\caption{Axioms for theorem \ref{thm:symmetric-monoidal-double-category}, part 3}
	\label{fig:symmetric-monoidal-double-category-3}
\end{figure}

\begin{figure}
	\begin{center} \begin{tikzpicture}[node distance=3cm, auto]
		\node (A) at (0, 0) {$\u (X_1 \otimes X_2)$}; \node (B) at (4, 0) {$\u (X_1) \otimes \u (X_2)$};
		\node (C) [below of=A] {$\u (X_2 \otimes X_1)$}; \node (D) [below of=B] {$\u (X_2) \otimes \u (X_1)$};
		\draw [->] (A) to node {$\mathfrak U_{X_1, X_2}$} (B); \draw [->] (A) to node {$\u (\sigma_{X_1, X_2})$} (C);
		\draw [->] (B) to node {$\sigma_{\u (X_1), \u (X_2)}$} (D); \draw [->] (C) to node {$\mathfrak U_{X_2, X_1}$} (D);
	\end{tikzpicture} \begin{tikzpicture}[node distance=3cm, auto]
		\node (A) at (0, 0) {$(\H_1 \otimes \H_2) \odot (\G_1 \otimes \G_2)$}; \node (B) at (6, 0) {$(\H_2 \otimes \H_1) \odot (\G_2 \otimes \G_1)$};
		\node (C) [below of=A] {$(\H_1 \odot \G_1) \otimes (\H_2 \odot \G_2)$}; \node (D) [below of=B] {$(\H_2 \odot \G_2) \otimes (\H_1 \odot \G_1)$};
		\draw [->] (A) to node {$\sigma_{\H_1, \H_2} \odot \sigma_{\G_1, \G_2}$} (B); \draw [->] (A) to node {$\mathfrak X_{\G_1, \G_2, \H_1, \H_2}$} (C);
		\draw [->] (B) to node {$\mathfrak X_{\G_2, \G_1, \H_2, \H_1}$} (D); \draw [->] (C) to node {$\sigma_{\H_1 \odot \G_1, \H_2 \odot \G_2}$} (D);
	\end{tikzpicture} \end{center}
	\caption{Axioms for theorem \ref{thm:symmetric-monoidal-double-category}, part 4}
	\label{fig:symmetric-monoidal-double-category-4}
\end{figure}

\begin{proof}[Proof (theorem \ref{thm:normal-form-theorem})]\label{prf:normal-form-theorem}
	We prove by induction on $n$ that the open game
	\[ \bigotimes_{i = 1}^n \D_{1, Y_i} : I \pto \left( \prod_{i = 1}^n Y_i, \R^n \right) \]
	is concretely given, up to unique natural isomorphism, as follows.
	Its set of strategy profiles is
	\[ \Sigma \left( \bigotimes_{i = 1}^n \D_{1, Y_i} \right) = \prod_{i = 1}^n Y_i \]
	The lens
	\[ \left( \bigotimes_{i = 1}^n \D_{1, Y_i} \right) (\sigma) : I \to \left( \prod_{i = 1}^n Y_i, \R^n \right) \]
	is the unique one with
	\[ \V \left( \left( \bigotimes_{i = 1}^n \D_{1, Y_i} \right) (\sigma) \right) = \sigma \]
	The best response relation
	\[ (\sigma, \sigma') \in \B_{\bigotimes_{i = 1}^n \D_{1, Y_i}} (h, k) \]
	holds iff, for all $1 \leq j \leq n$ and all $y_j : Y_j$,
	\[ k (\sigma'_j, \sigma_{-j})_j \geq k (y_j, \sigma_{-j})_j \]
	
	The result follows, since the set of strategy profiles of $\bigotimes_{i = 1}^n \D_{1, Y_i}$ is the set of strategy profiles of the normal form game, and fixpoints of $\B_{\bigotimes_{i = 1}^n \D_{1, Y_i}} (*, k)$ are Nash equilibria by definition.
	
	The set of strategy profiles follows since $\Sigma : \Game_v \to \Set$ is a symmetric monoidal functor.
	For the lens we have
	\[ \V \left( \left( \bigotimes_{i = 1}^n \D_{1, Y_i} \right) (\sigma) \right) = \V \left( \bigotimes_{i = 1}^n \D_{1, Y_i} (\sigma_i) \right) = \prod_{i = 1}^n \V (\D_{1, Y_i} (\sigma_i)) = \prod_{i = 1}^n \sigma_i = \sigma \]
	
	We prove the claim about best responses by induction on $n$.
	When $n = 1$, we have by definition that $(\sigma_1, \sigma'_1) \in \B_{1, Y_1} (*, k)$ iff $k (\sigma'_1) \geq k (y_1)$ for all $y_1 : Y_1$.
	This has the required form because $k (y_1, \sigma_{-1})_1 = k (y_1)$.
	
	For the inductive step, by definition
	\[ (\sigma, \sigma') \in \B_{\bigotimes_{i = 1}^{n + 1} \mathcal D_{1, Y_i}} (*, k) \]
	iff
	\[ (\sigma_{-(n + 1)}, \sigma'_{-(n + 1)}) \in \B_{\bigotimes_{i = 1}^n \mathcal D_{1, Y_i}} (L (\mathcal D_{1, Y_{n + 1}} (\sigma_{n + 1})) (*, k)) \]
	and
	\[ (\sigma_{n + 1}, \sigma'_{n + 1}) \in \B_{\mathcal D_{1, Y_{n + 1}}} \left( R \left( \left( \bigotimes_{i = 1}^n \mathcal D_{1, Y_i} \right) (\sigma_{-(n + 1)}) \right) (*, k) \right) \]
	
	Writing
	\[ L (\mathcal D_{1, Y_{n + 1}} (\sigma_{n + 1})) (*, k) = (*, k_L) \]
	and
	\[ R \left( \left( \bigotimes_{i = 1}^n \mathcal D_{1, Y_i} \right) (\sigma_{-(n + 1)}) \right) (*, k) = (*, k_R) \]
	we can directly calculate
	\[ k_L (y_{-(n + 1)}) = k (y_{-(n + 1)}, \sigma_{n + 1})_{-(n + 1)} \]
	and
	\[ k_R (y_{n + 1}) = k (\sigma_{-(n + 1)}, y_{n + 1})_{n + 1} \]
	
	By the inductive hypothesis, the first condition
	\[ (\sigma_{-(n + 1)}, \sigma'_{-(n + 1)}) \in \B_{\bigotimes_{i = 1}^n \D_{1, Y_i}} (*, k_L) \]
	holds iff for all $1 \leq j \leq n$ and all $y_j : Y_j$,
	\[ k_L  ((\sigma'_{-(n + 1)})_j, (\sigma_{-(n + 1)})_{-j})_j \geq k_L (y_j, (\sigma_{-(n + 1)})_{-j})_j \]
	After substituting $k_L$ this is equivalent to
	\[ (k ((\sigma'_{-(n + 1)})_j, (\sigma_{-(n + 1)})_{-j}, \sigma_{n + 1})_{-(n + 1)})_j 	\geq (k (y_j, (\sigma_{-(n + 1)})_{-j}, \sigma_{n + 1})_{-(n + 1)})_j \]
	and hence, after composing projectors, to
	\[ k (\sigma'_j, \sigma_{-j})_j \geq k (y_j, \sigma_{-j})_j \]
	
	The second condition
	\[ (\sigma_{n + 1}, \sigma'_{n + 1}) \in \B_{\D_{1, Y_{n + 1}}} (*, k_R) \]
	holds by definition iff for all $y_{n + 1} : Y_{n + 1}$,
	\[ k_R (\sigma'_{n + 1}) \geq k_R (y_{n + 1}) \]
	which is
	\[ k (\sigma_{-(n + 1)}, \sigma'_{n + 1})_{n + 1} \geq k (\sigma_{-(n + 1)}, y_{n + 1})_{n + 1} \]
	
	Putting these two conditions together, we have the inductive hypothesis for $n + 1$.
\end{proof}

\begin{proof}[Proof (theorem \ref{thm:extensive-form-theorem})]\label{prf:extensive-form-theorem}
	Let
	\[ \G_{X_1, \ldots, X_n} = \bigodot_{i = 1}^n \D^\Delta_{X_1, \ldots, X_i} : I \pto \left( \prod_{i = 1}^n X_i, \R^n \right) \]
	We prove that $\G_{X_1, \ldots, X_n}$ is given explicitly as follows, up to globular isomorphism.
	The set of strategy profiles is
	\[ \Sigma (\G_{X_1, \ldots, X_n}) = \prod_{i = 1}^n \left( \prod_{j = 1}^{i - 1} X_j \to X_i \right) \]
	For $\sigma : \Sigma (\G_{X_1, \ldots, X_n})$, the lens $\G_{X_1, \ldots, X_n} (\sigma)$ is the unique one with
	\[ \V (\G_{X_1, \ldots, X_n} (\sigma)) (*) = v^\sigma \]
	For $k : \prod_{i = 1}^n X_i \to \R^n$, the best response relation
	\[ (\sigma, \sigma') \in \B_{\G_{X_1, \ldots, X_n}} (*, k) \]
	holds iff for all $1 \leq i \leq n$ and all $x_i : X_i$,
	\[ k \left( v^\sigma_{(v^\sigma)_1^{i - 1}, \sigma'_i ((v^\sigma)_1^{i - 1})} \right)_i \geq k \left( v^\sigma_{(v^\sigma)_1^{i - 1}, x_i} \right)_i \]

	We prove the first claim by induction on $n$.
	It is straightforward to check the base case $\G_{X_1} = \D^\Delta_{X_1} : I \pto (X_1, \R)$.
	
	The inductive step is
	\[ \G_{X_1, \ldots, X_n, X_{n + 1}} = \D^\Delta_{X_1, \ldots, X_n, X_{n + 1}} \odot \G_{X_1, \ldots, X_n} \]
	Its set of strategies is
	\begin{align*}
		\Sigma (\G_{X_1, \ldots, X_n, X_{n + 1}}) &= \Sigma (\G_{X_1, \ldots, X_n}) \times \Sigma (\D^\Delta_{X_1, \ldots, X_n, X_{n + 1}}) \\
		&= \prod_{i = 1}^n \left( \prod_{j = 1}^{i - 1} X_j \to Xi \right) \times \left( \prod_{j = 1}^n X_j \to X_{n + 1} \right) \\
		&= \prod_{i = 1}^{n + 1} \left( \prod_{j = 1}^{i - 1} X_j \to X_i \right)
	\end{align*}
	The lens is
	\begin{align*}
		\V (\G_{X_1, \ldots, X_n, X_{n + 1}} (\sigma)) (*) &= (\V (\D^\Delta_{X_1, \ldots, X_n, X_{n + 1}} (\sigma_{n + 1})) \circ \V (\G_{X_1, \ldots, X_n} (\sigma_{-(n + 1)}))) (*) \\
		&= \V (\D^\Delta_{X_1, \ldots, X_n, X_{n + 1}} (\sigma_{n + 1})) (v^{\sigma_{-(n + 1)}}) \\
		&= (v^{\sigma_{-(n + 1)}}, \sigma_{n + 1} (v^{\sigma_{-(n + 1)}})) \\
		&= v^\sigma
	\end{align*}
	
	For $k : \prod_{i = 1}^{n + 1} X_i \to \R^{n + 1}$, the best response relation
	\[ (\sigma, \sigma') \in \B_{\G_{X_1, \ldots, X_n, X_{n + 1}}} (*, k) \]
	holds iff
	\begin{equation}\label{eqn:condition-1}
		(\sigma_{-(n + 1)}, \sigma'_{-(n + 1)}) \in \B_{\G_{X_1, \ldots, X_n}} (*, \K (\D^\Delta_{X_1, \ldots, X_n, X_{n + 1}} (\sigma_{n + 1})) (k))
	\end{equation}
	and
	\begin{equation}\label{eqn:condition-2}
		(\sigma_{n + 1}, \sigma'_{n + 1}) \in \B_{\D^\Delta_{X_1, \ldots, X_n, X_{n + 1}}} (\V (\G_{X_1, \ldots, X_n} (\sigma_{-(n + 1)})) (*), k)
	\end{equation}
	
	In condition \ref{eqn:condition-1}, the continuation is
	\[ \K (\D^\Delta_{X_1, \ldots, X_n, X_{n + 1}} (\sigma_{n + 1})) (k) : \prod_{i = 1}^n X_i \to \R^n \]
	\[ x \mapsto u_{\D^\Delta_{X_1, \ldots, X_n, X_{n + 1}} (\sigma_{n + 1})} (x, k (v_{\D^\Delta_{X_1, \ldots, X_n, X_{n + 1}} (\sigma_{n + 1})} (x))) = k (x, \sigma_{n + 1} (x))_{-(n + 1)} \]
	By the inductive hypothesis, this relation holds iff for all $1 \leq i \leq n$ and $x_i : X_i$,
	\[ \K (\D^\Delta_{X_1, \ldots, X_n, X_{n + 1}} (\sigma_{n + 1})) (k) \left( v^\sigma_{(v^\sigma)_1^{i - 1}, \sigma'_i ((v^\sigma)_1^{i - 1})} \right)_i \geq \K (\D^\Delta_{X_1, \ldots, X_n, X_{n + 1}} (\sigma_{n + 1})) (k) \left( v^\sigma_{(v^\sigma)_1^{i - 1}, x_i} \right) \]
	The left hand side of this inequation is
	\begin{align*}
		\K (\D^\Delta_{X_1, \ldots, X_n, X_{n + 1}} (\sigma_{n + 1})) (k) \left( v^\sigma_{(v^\sigma)_1^{i - 1}, \sigma'_i ((v^\sigma)_1^{i - 1})} \right)_i &= \left( k \left( v^\sigma_{(v^\sigma)_1^{i - 1}, \sigma'_i ((v^\sigma)_1^{i - 1})}, \sigma_{n + 1} \left( v^\sigma_{(v^\sigma)_1^{i - 1}, \sigma'_i ((v^\sigma)_1^{i - 1})} \right) \right)_{-(n + 1)} \right)_i \\
		&= k \left( v^\sigma_{(v^\sigma)_1^{i - 1}, \sigma'_i ((v^\sigma)_1^{i - 1})} \right)_i
	\end{align*}
	and similarly the right hand side is
	\[ \K (\D^\Delta_{X_1, \ldots, X_n, X_{n + 1}} (\sigma_{n + 1})) (k) \left( v^\sigma_{(v^\sigma)_1^{i - 1}, x_i} \right) = k \left( v^\sigma_{(v^\sigma)_1^{i - 1}, x_i} \right)_i \]
	
	In condition \ref{eqn:condition-2}, by the inductive hypothesis the history is
	\[ \V (\G_{X_1, \ldots, X_n} (\sigma_{-(n + 1)})) (*) = v^{\sigma_{-(n + 1)}} \]
	By definition of $\D^\Delta_{X_1, \ldots, X_n, X_{n + 1}}$, condition \ref{eqn:condition-2} holds iff
	\[ k (v^{\sigma_{-(n + 1)}}, \sigma'_{n + 1} (v^{\sigma_{-(n + 1)}}))_{n + 1} \geq k (v^{\sigma_{-(n + 1)}}, x_{n + 1})_{n + 1} \]
	for all $x_{n + 1} : X_{n + 1}$.
	The left hand side of this inequation is
	\[ k (v^{\sigma_{-(n + 1)}}, \sigma'_{n + 1} (v^{\sigma_{-(n + 1)}}))_{n + 1} = k (v^\sigma_{(v^\sigma)_1^n, \sigma'_{n + 1} ((v^\sigma)_1^n)})_{n + 1} \]
	and similarly the right hand side is
	\[ k (v^{\sigma_{-(n + 1)}}, x_{n + 1})_{n + 1} = k (v^\sigma_{(v^\sigma)_1^n, x_{n + 1}})_{n + 1} \]
	
	Putting together conditions \ref{eqn:condition-1} and \ref{eqn:condition-2}, we obtain the inductive hypothesis for $n + 1$.
	
	Next we prove by induction on $i$ that there is a bijective correspondence between:
	\begin{itemize}
		\item States $\alpha$ of
		\[ \bigodot_{j = n - i + 1}^n \D^\Delta_{X_1, \ldots, X_j} : \left( \prod_{j = 1}^{n - i} X_j, \R^{n - i} \right) \pto \left( \prod_{j = 1}^n X_j, \R^n \right) \]
		that are of the form
		\[ \alpha = \bigodot_{j = n - i + 1}^n \alpha_j \]
		where each $\alpha_j$ is a state of $\D^\Delta_{X_1, \ldots, X_j}$
		
		\item Strategy profiles
		\[ \sigma_{n - i + 1}^n : \prod_{j = n - i + 1}^n \left( \prod_{l = 1}^{j - 1} X_l \to X_j \right) \]
		with the property that for all $j \geq n - i + 1$, all subgames $x : \prod_{l = 1}^{j - 1} X_l$ and all deviations $x_j : X_j$,
		\[ k (v^\sigma_{x, \sigma_j (x)})_j \geq k (v^\sigma_{x, x_j})_j \]
	\end{itemize}

	In the base case $i = 1$, we have immediately that states $\alpha$ of $\D^\Delta_{X_1, \ldots, X_n}$ over $k$ are in bijection with strategies $\sigma_n : \prod_{l = 1}^{n - 1} X_l \to X_n$ with the property that for all subgames $x : \prod_{j = 1}^{n - 1} X_j$ and all deviations $x_n : X_n$,
	\[ k (v^\sigma_{x, \sigma_n (x)})_n = k (x, \sigma_n (x))_n \geq k (x, x_n)_n = k (v^\sigma_{x, x_n})_n \]
	
	For the inductive step, a $\odot$-separable state $\alpha$ of $\bigodot_{j = n - i}^n \D^\Delta_{X_1, \ldots, X_j}$ is of the form $\alpha = \alpha' \odot \alpha_{n - i}$, where $\alpha_{n - i}$ is a state of $\D^\Delta_{X_1, \ldots, X_{n - i}}$ and $\alpha'$ is a $\odot$-separable state of $\bigodot_{j = n - i + 1}^n \D^\Delta_{X_1, \ldots, X_j}$.
	This situation is depicted in figure \ref{fig:extensive-form-lemma}.
	By the inductive hypothesis, the latter are in bijection with strategy profiles
	\[ \sigma : \prod_{j = n - i + 1}^n \left( \prod_{l = 1}^{j - 1} X_l \to X_j \right) \]
	with the property that for all $j \geq n - i + 1$, all subgames $x : \prod_{l = 1}^{j - 1} X_l$ and all deviations $x_j : X_j$,
	\[ k (v^\sigma_{x, \sigma_j (x)})_j \geq k (v^\sigma_{x, x_j})_j \]
	
	The state $\alpha_{n - i}$ is over the continuation
	\[ \K \left( \left( \bigodot_{j = n - i + 1}^n \D^\Delta_{X_1, \ldots, X_j} \right) (\sigma) \right) (k) : \prod_{j = 1}^{n - i} X_j \to \R^{n - i} \]
	We call this continuation $k'$.
	It is given by
	\[ k' (x) = \left( k \left( v^{\sigma'}_x \right) \right)_1^{n - i} \]
	Thus states $\alpha_{n - i}$ with the property that the composition $\alpha' \odot \alpha_{n - i}$ is well-defined are in bijection with strategies
	\[ \sigma_{n - i} : \Sigma \left( \D^\Delta_{X_1, \ldots, X_{n - i}} \right) = \prod_{j = 1}^{n - i - 1} X_j \to X_{n - i} \]
	with the property that for all histories
	\[ h : \prod_{j = 1}^{n - i - 1} X_j \cong \V \left( \prod_{j = 1}^{n - i - 1} X_j, \R^{n - i - 1} \right) \]
	and all deviations $x_{n - i} : X_{n - i}$,
	\[  k' (h, \sigma_{n - i} (h))_{n - i} \geq k' (h, x_{n - i})_{n - i} \]
	that is to say
	\[ k \left( v^{\sigma'}_{h, \sigma_{n - i} (h)} \right)_{n - i} \geq k \left( v^{\sigma'}_{h, x_{n - i}} \right)_{n - i} \]
	
	Putting these together gives equivalence to the inductive hypothesis for $i + 1$.
	
	The second claim of the theorem follows by taking $i = n$.
\end{proof}

\begin{figure}
	\begin{center} \begin{tikzpicture}[node distance=6cm, auto]
		\node (A) {$I$}; \node (B) [right of=A] {$I$}; \node (C) [right of=B] {$I$};
		\node (D) at (0, -3) {$\displaystyle \left( \prod_{j = 1}^{n - i - 1} X_j, \R^{n - i - 1} \right)$};
		\node (E) [right of=D] {$\displaystyle \left( \prod_{j = 1}^{n - i} X_j, \R^{n - i} \right)$};
		\node (F) [right of=E] {$\displaystyle \left( \prod_{j = 1}^n X_j, \R^n \right)$};
		\draw [-|->] (A) to node [above] {$\u (I)$} node [below] {$1$} (B); \draw [-|->] (B) to node [above] {$\u (I)$} node [below] {$1$} (C);
		\draw [-|->] (D) to node [above] {$\D^\Delta_{X_1, \ldots, X_{n - i}}$} node [below] {$\displaystyle \prod_{j = 1}^{n - i - 1} X_j \to X_{n - i}$} (E);
		\draw [-|->] (E) to node [above] {$\displaystyle \bigodot_{j = n - i + 1}^n \D^\Delta_{X_1, \ldots, X_j}$} node [below] {$\displaystyle \prod_{j = n - i + 1}^n \left( \prod_{l = 1}^{j - 1} X_l \to X_j \right)$} (F);
		\draw [->] (F) to node [right] {$k$} (C); \draw [->] (E) to node [right] {$k'$} (B); \draw [->] (D) to (A);
		\draw [->] (3, -.5) to node [right] {$\sigma_{n - i}$} (3, -2.3);
		\draw [->] (9, -.5) to node [right] {$\sigma'$} (9, -1.7);
	\end{tikzpicture} \end{center}
	\caption{Inductive step of theorem \ref{thm:extensive-form-theorem}}
	\label{fig:extensive-form-lemma}
\end{figure}

\begin{proof}[Proof (proposition \ref{prop:product-games})]\label{prf:product-games}
	We first prove that the projections $\pi_j$ satisfy the axioms of a morphism of open games.
	For a strategy profile $\sigma : \prod_{i : I} \Sigma (\G_i)$ the diagram
	\begin{center} \begin{tikzpicture}[node distance=3cm, auto]
		\node (A) at (0, 0) {$\displaystyle \left( \coprod_{i : I} X_i, S \right)$}; \node (B) at (5, 0) {$\displaystyle \left( \coprod_{i : I} Y_i, R \right)$};
		\node (C) [below of=A] {$(X_j, S)$}; \node (D) [below of=B] {$(Y_j, R)$};
		\draw [->] (A) to node {$\coprod_{i : I} \G_i (\sigma_i)$} (B); \draw [->] (C) to node {$\G_j (\sigma_j)$} (D);
		\draw [->] (C) to node [right] {$(\iota_j, S)$} (A); \draw [->] (D) to node [right] {$(\iota_j, R)$} (B);
	\end{tikzpicture} \end{center}
	commutes by the universal property of the coproduct in $\Lens$.
	The second axiom of morphisms holds directly by definition.
	
	Suppose we have an open game $\H : \Phi \pto \Psi$ and a family of morphisms $\alpha_i : \H \pto \G_i$.
	We have unique choices for the $\s$, $\t$ and $\Sigma$-components of the universal morphism, namely
	\begin{center} \begin{tikzpicture}[node distance=3cm, auto]
		\node (A) at (0, 0) {$\Phi$}; \node (B) at (5, 0) {$\Psi$};
		\node (C) [below of=A] {$\displaystyle \left( \coprod_{i : I} X_i, S \right)$}; \node (D) [below of=B] {$\displaystyle \left( \coprod_{i : I} Y_i, R \right)$};
		\draw [-|->] (A) to node [above] {$\H$} node [below] {$\Sigma (\H)$} (B);
		\draw [-|->] (C) to node [above] {$\displaystyle \prod_{i : I} \G_i$} node [below] {$\displaystyle \prod_{i : I} \Sigma (\G_i)$} (D);
		\draw [->] (C) to node [right] {$[\s (\alpha_i)]_{i : I}$} (A); \draw [->] (D) to node [right] {$[\t (\alpha_i)]_{i : I}$} (B);
		\draw [->] (2.5, -.5) to node {$\left< \Sigma (\alpha_i) \right>_{i : I}$} (2.5, -2);
	\end{tikzpicture} \end{center}
	It suffices to prove that this does indeed define a morphism of open games.
	For a strategy profile $\sigma : \Sigma (\H)$, commutativity of
	\begin{center} \begin{tikzpicture}[node distance=3cm, auto]
		\node (A) at (0, 0) {$\Phi$}; \node (B) at (5, 0) {$\Psi$};
		\node (C) [below of=A] {$\displaystyle \left( \coprod_{i : I} X_i, S \right)$}; \node (D) [below of=B] {$\displaystyle \left( \coprod_{i : I} Y_i, R \right)$};
		\draw [->] (A) to node {$\H (\sigma)$} (B); 	\draw [->] (C) to node {$\displaystyle \coprod_{i : I} \G_i (\Sigma (\alpha_i) (\sigma))$} (D);
		\draw [->] (C) to node [right] {$[\s (\alpha_i)]_{i : I}$} (A); \draw [->] (D) to node [right] {$[\t (\alpha_i)]_{i : I}$} (B);
	\end{tikzpicture} \end{center}
	follows from the individual $\alpha_i$ being morphisms.
	
	For the second axiom, let $\sigma, \sigma' : \Sigma (\H)$, $\iota_j (h) : \coprod_{i : I} X_i$ and $k : \K (\Psi)$.
	Suppose that
	\[ (\sigma, \sigma') \in \B_\H (\V (\s (\alpha_j)) (h), k) \]
	Since
	\[ \V ([\s (\alpha_i)]_{i : I}) (\iota_j (h)) = \V (\s (\alpha_j)) (h) \]
	by the second axiom of $\alpha_j$ we have
	\[ (\Sigma (\alpha_j) (\sigma), \Sigma (\alpha_j) (\sigma')) \in \B_{\G_j} (h, \K (\t (\alpha_j)) (k)) \]
	Since the diagram
	\begin{center} \begin{tikzpicture}[node distance=3cm, auto]
		\node (A) at (0, 0) {$Y_j$}; \node (B) at (5, 0) {$R$}; \node (C) [below of=A] {$\displaystyle \coprod_{i : I} Y_i$};
		\draw [->] (A) to node {$\K (\t (\alpha_j)) (k)$} (B); \draw [->] (A) to node {$\iota_j$} (C);
		\draw [->] (C) to node [right=10pt] {$\K ([\t (\alpha_i)]_{i : I}) (k)$} (B);
	\end{tikzpicture} \end{center}
	commutes, it follows that
	\[ (\left< \Sigma (\alpha_i) \right>_{i : I} (\sigma), \left< \Sigma (\alpha_i) \right>_{i : I} (\sigma')) \in \B_{\prod_{i : I} \G_i} (\iota_j (h), \K ([\t (\alpha_i)]_{i : I}) (k)) \]
	as required.
\end{proof}

\end{document}